\documentclass[floats,10pt,onecolumn]{revtex4}
\usepackage{enumerate}
\usepackage{xcolor} 
\definecolor{DarkGray}{rgb}{0.1,0.1,0.5}
\usepackage[colorlinks=true,breaklinks, linkcolor= DarkGray,citecolor= DarkGray,urlcolor= DarkGray]{hyperref}
\usepackage{url}
\usepackage{verbatim}
\usepackage{amsmath}  
\usepackage{latexsym}
\usepackage{revsymb}
\usepackage{epsfig} 
\usepackage{pstricks}
\usepackage{ifthen}
\usepackage{mathrsfs}

\usepackage{natbib}
\usepackage{amsfonts}
\usepackage{amsmath}
\usepackage{amssymb}
\usepackage{subfigure}
\usepackage{amsthm}
\usepackage{graphicx}
\newcommand{\appref}[1]{\hyperref[#1]{{Appendix~\ref*{#1}}}}
\newcommand{\be}{\begin{eqnarray} \begin{aligned}}
\newcommand{\ee}{\end{aligned} \end{eqnarray} }
\newcommand{\benn}{\begin{eqnarray*} \begin{aligned}}
\newcommand{\eenn}{\end{aligned} \end{eqnarray*} }

\newcommand*{\textfrac}[2]{{{#1}/{#2}}}

\newcommand*{\cB}{\mathcal{B}}

\newcommand*{\cE}{\mathcal{E}}
\newcommand*{\cF}{\mathcal{F}}

\newcommand*{\cI}{\mathcal{I}}

\newcommand*{\cM}{\mathcal{M}}
\newcommand*{\cN}{\mathscr{N}}

\newcommand*{\cO}{\mathcal{O}}
\newcommand*{\cP}{\mathcal{P}}

\newcommand*{\tr}{\mathop{\mathrm{tr}}\nolimits}
\newcommand*{\sbin}{\{0,1\}}

\newcommand{\bc}{\begin{center}}
\newcommand{\ec}{\end{center}}

\newcommand{\id}{\mathbb{I}}

\newtheorem{theorem}{Theorem}[section]
\newtheorem{lemma}[theorem]{Lemma}

\newtheorem{corollary}[theorem]{Corollary}

\newcommand{\myacknowledgments}{\begin{center}{\bf Acknowledgments}\end{center}\par}

\usepackage{amsfonts}

\def\id{\mathbb{I}}

\def\01{\{0,1\}}

\newcommand{\ket}[1]{|#1\rangle}
\newcommand{\bra}[1]{\langle#1|}

\newcommand{\proj}[1]{|#1\rangle\langle#1|}

\newcounter{protoCount}
\newcounter{protoList}
\newsavebox{\tmpbox}
\newlength{\protobox}
\newenvironment{protocol}[3]{
\bigskip
\addtocounter{protoCount}{1}
\noindent \begin{lrbox}{\tmpbox}
\setlength{\protobox}{\textwidth}
\addtolength{\protobox}{-0.5cm}
\begin{minipage}[c]{\protobox}
\begin{bfseries}Protocol #1: #2\end{bfseries}
\ifthenelse{\equal{#3}{\empty}}{}{\\ #3}
\begin{list}{\begin{bfseries}\arabic{protoList}:\end{bfseries}}
{\usecounter{protoList}}
}{
\end{list}
\end{minipage}\end{lrbox}
\fbox{\usebox{\tmpbox}}
\bigskip
}

\newcommand*{\myvec}[1]{{\bf #1}}

\usepackage{amsmath}
\usepackage{amssymb}

\newcommand*{\unitaryinst}{\raisebox{-76.446pt}{\epsfig{file=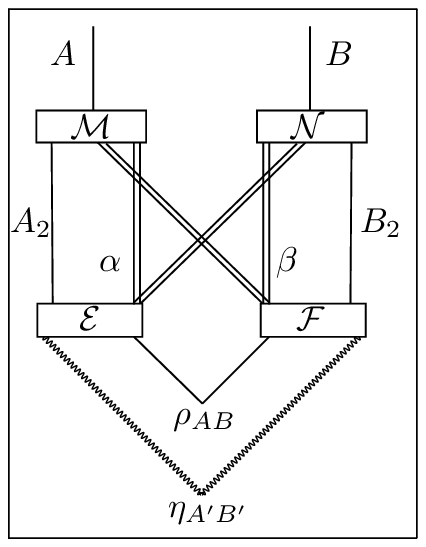,clip=}}}

\newcommand*{\measureinst}{\raisebox{-76.446pt}{\epsfig{file=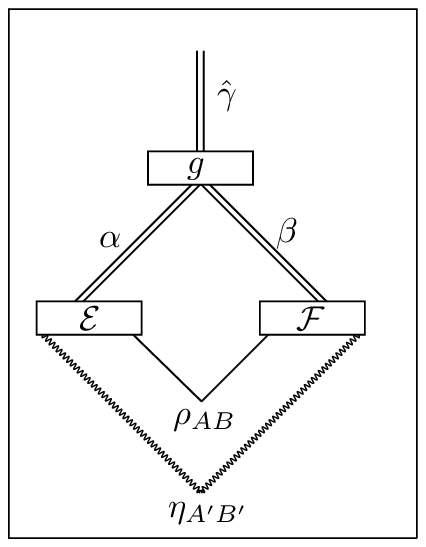,clip=}}}

\newcommand*{\measurement}{\raisebox{-76.446pt}{\epsfig{file=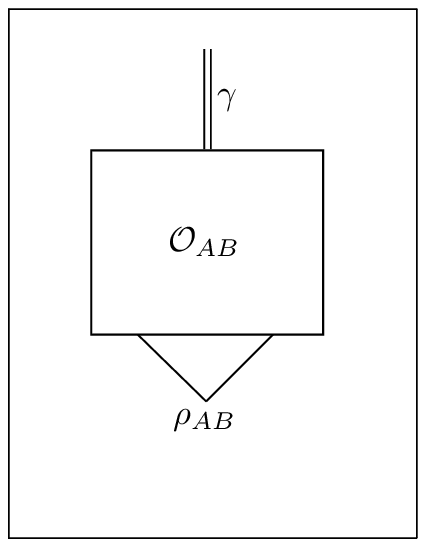,clip=}}}

\newcommand*{\unitary}{\raisebox{-76.446pt}{\epsfig{file=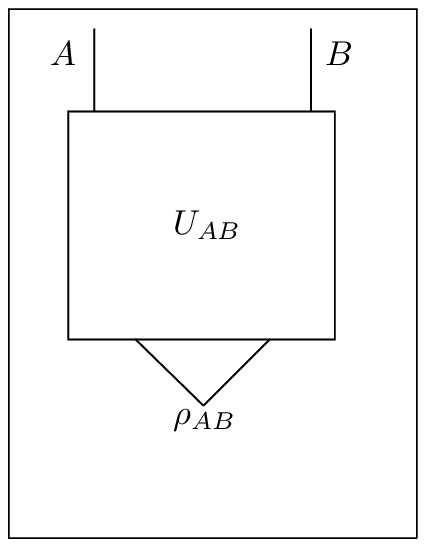,clip=}}}

\begin{document}

\title{Simplified instantaneous non-local quantum computation with applications to position-based cryptography}
\author{Salman Beigi$^{1,2}$ and Robert K\"onig$^{1,3}$}
\address{$^1$ Institute for Quantum Information, Caltech, Pasadena CA 91125, USA}
\address{$^2$ School of Mathematics,
 Institute for Research in Fundamental Sciences (IPM),
 Tehran, Iran}
\address{$^3$ IBM T.J. Watson Research Center, Yorktown Heights, NY 10598, USA}

\begin{abstract}
Instantaneous measurements of non-local observables between space-like separated regions can be performed without violating causality. This feat relies on the use of  entanglement. Here we propose novel protocols for this task and the related problem of  multipartite quantum computation with local operations and a single round of classical communication. Compared to previously known techniques, our protocols reduce the entanglement consumption by an exponential amount.  We also prove a linear lower bound on the amount of entanglement required for the implementation of a certain non-local measurement.

These results relate to position-based cryptography:  an amount of entanglement scaling exponentially in the number of communicated qubits is sufficient to render any such scheme insecure. Furthermore, we show that certain schemes are secure under the assumption that the adversary has less entanglement than a given linear bound and  is restricted to classical communication.
\end{abstract}
\maketitle

It is remarkable that the axioms of quantum mechanics are compatible with the severe restrictions imposed by relativistic causality. From the early days of quantum mechanics, this miraculous fact has repeatedly been called into question  and has been a source of great controversy. Arguably the most well-known debate of this kind has centered around the Einstein-Podolsky-Rosen (EPR) paradox~\cite{epr35}, which deals with the non-local correlations arising from a bilocal measurement of an entangled state.

  A similar discussion originated from concerns about the compatibility of the measurement process with relativistic quantum field theory. In~1931, Landau and Peierls~\cite{landaupeierls31} showed that the electromagnetic field strength cannot be accurately measured  by means of point-like test charges: an uncertainty relation implies large fluctuations in their positions, which in turn leads to the emission of radiation strongly influencing the field elsewhere. From this result,  Landau and Peierls concluded  that, quite generally, the standard measurement prescription of quantum mechanics does not apply in a relativistic setting.   

Such difficulties reconciling quantum measurements with causality are strikingly apparent when considering  non-separable bilocal measurements (e.g., a von Neumann measurement in the Bell basis of two qubits). The instant collapse of the wavefunction induced by such measurements allows to signal instantly between space-like separated regions (see e.g.,~\cite{clarketal10} for explicit examples). This may suggest that certain non-local observables are not measurable at a well-defined time in a fixed Lorentz frame. As a consequence, one may be led to believe that their expectation values do not carry any physical meaning. As with the EPR paradox, entanglement is at the root of this apparent causal restriction on the set of physically allowed observables, yet in this case, it is the observable (or measurement) instead of the state which  is entangled. Somewhat ironically, entanglement also figures prominently in the resolution of this issue.

Landau and Peierls' far-reaching conclusions were soon challenged by Bohr and Rosenfeld~\cite{bohrrosenfeld33}, who showed how the electromagnetic field strength can be measured using spatially extended charge distributions instead of point-like test charges.  Much later, Aharonov and Albert~\cite{aharonovalbert80,aharonovalbert81} proposed a way of measuring certain non-local variables in a manner consistent with causality. This involves the use of prior shared entanglement. A series of subsequent works~\cite{ahaa84,ahab84,ahaavaid86,popescuvaid94,groismanvaidman01} characterizing and extending this kind of instantaneous measurements culminated in a scheme by Vaidman~\cite{vaidman03}, which allows to instantaneously measure any non-local observable. This disproves the stated existence of causality restrictions on non-local measurements.

\begin{figure}
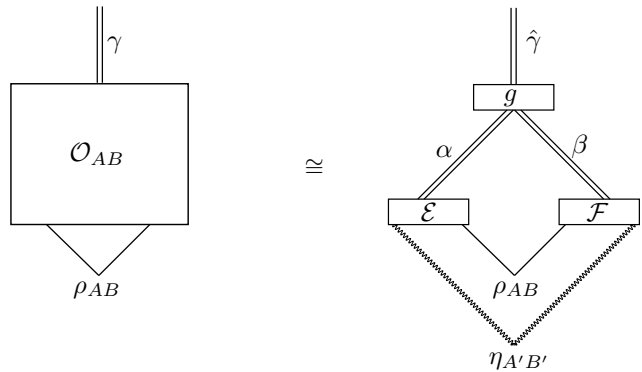

\begin{center}
\measurement\qquad $\cong$\qquad \measureinst
\end{center}
\caption{Instantaneous measurement of a non-local observable~$\cO_{AB}$:
Alice and Bob share, in addition to the state~$\rho_{AB}$ to be measured,  an auxiliary entangled state~$\eta_{A'B'}$ (indicated by the wiggly line).
They perform local  measurements~$\cE=\{E^\alpha_{AA'}\}_\alpha$ and $\cF=\{F^\beta_{BB'}\}_\beta$, respectively. Charlie computes a function $\hat{\gamma}=g(\alpha,\beta)$ of their measurement results. The measurements and the post-processing function are chosen in such a way that this simulates the measurement of~$\rho_{AB}$ with the non-local POVM~$\cO_{AB}=\{O^\gamma_{AB}\}_\gamma$. \label{fig:instmeasurement}}
\end{figure}

The basic achievement of such measurement schemes is illustrated in Fig.~\ref{fig:instmeasurement}. Two space-like separated observers Alice (A) and Bob (B) sharing a bipartite system~$AB$ aim to determine a certain non-local property of their joint state $\rho_{AB}=\rho_{AB}(t_0)$ at a specific time~$t_0$. This property is described by a non-local  POVM  with operators~$\cO=\{O^\gamma_{AB}\}_{\gamma}$, and their goal is to sample from the probability distribution 
\begin{equation} 
p(\gamma)=\tr(O^\gamma_{AB}\rho_{AB})\ .\label{eq:distributionnonlocal}
\end{equation}
 The apparent causality problem arises when  the operators constituting the measurement are non-separable.   Vaidman shows that this sampling problem can be solved as follows (see Fig.~\ref{fig:instmeasurement}):
\begin{enumerate}
\item
 First, Alice and Bob apply local measurements at time~$t_0$ to their shared state $\rho_{AB}\otimes \eta_{A'B'}$. Here $\eta_{A'B'}$ is an auxiliary shared entangled state. They obtain outcomes $\alpha$ and $\beta$, respectively.

\item
Alice sends $\alpha$ and Bob sends $\beta$ to Charlie (a point~$C$ in the intersection of the causal cones of $A$ and $B$).
\item
 At a later time~$t>t_0$ after reception of  the measurement outcomes~$(\alpha,\beta)$, Charlie computes a value~$\hat{\gamma}$ by applying some function~$g$ to the pair $(\alpha,\beta)$.
\end{enumerate}
 In other words, all measurements are local and instantaneous and performed at time~$t_0$. The output is a value~$\hat{\gamma}$ which is  computed at a time~$t>t_0$ at~$C$, but is supposed to pertain to a non-local measurement of~$\rho_{AB}(t_0)$ at time~$t_0$.

For any POVM $\{O^\gamma_{AB}\}_\gamma$, Vaidman constructs local measurements $\cE=\{E^\alpha_{AA'}\}_\alpha$  and $\cF=\{F^\beta_{BB'}\}_\beta$ and a postprocessing function $g$ such that the distribution
\begin{equation*}
\hat{p}(\hat{\gamma}) =\sum_{(\alpha,\beta): g(\alpha,\beta)=\hat{\gamma}}\tr\left((E_{AA'}^\alpha\otimes F_{BB'}^\beta)(\rho_{AB}\otimes\eta_{A'B'})\right)
\end{equation*}
is close to the distribution~\eqref{eq:distributionnonlocal}. Entanglement therefore allows to estimate expectation values of non-local observables at a specific instant in time without violating causality.

Beyond realizing the statistics of non-local measurements, Vaidman's scheme 
also provides an instantaneous implementation of non-local operations. Here the goal is to apply a (non-local) unitary $U_{AB}$ to the joint state~$\rho_{AB}$ of Alice and Bob. They are restricted to applying local operations and only a single round of simultaneously passed classical communication, see Fig.~\ref{fig:simultaneouscomputation}. This very limited form of interaction acts as a non-signaling constraint and  makes this task non-trivial for general unitaries. Again, Vaidman's techniques  demonstrate that prior shared entanglement allows to implement such non-local operations. This fact was first recognized in~\cite{buhrmanetal10}, where it was used to address a problem in cryptography.
\begin{figure}
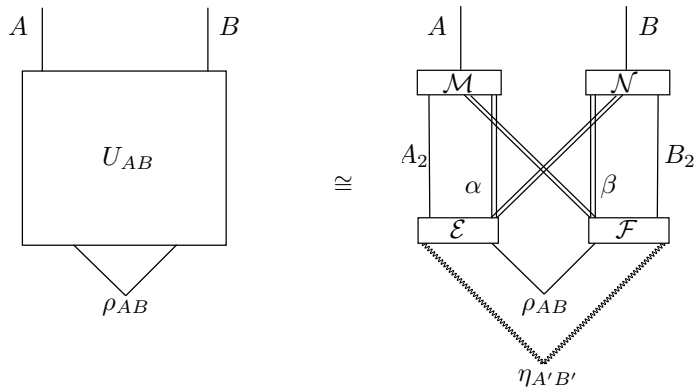
 
\begin{center}
\unitary\qquad $\cong$\qquad \unitaryinst
\end{center}
\caption{Instantaneous implementation of a non-local unitary~$U_{AB}$ on a bipartite state~$\rho_{AB}$ using the shared entangled state~$\eta_{A'B'}$.
Alice and Bob perform local (partial) measurements~$\cE=\{E^\alpha_{A_1}\}_\alpha$ and $\cF=\{F^\beta_{B_1}\}_\beta$ where $A_1A_2=AA'$ and $B_1B_2=BB'$, respectively. This results in the  residual state~$\rho^{\alpha,\beta}_{A_2B_2}:=\tr_{A_1B_1} \left( \left(\id_{A_2\otimes B_2}\otimes E^\alpha_{A_1}\otimes F^\beta_{B_1} \right)(\rho_{AB}\otimes \eta_{A'B'}) \right) /p(\alpha,\beta)$ with probability~$p(\alpha,\beta) =\tr \left(\left(\id_{A_2\otimes B_2}\otimes E^\alpha_{A_1}\otimes F^\beta_{B_1}\right)(\rho_{AB}\otimes \eta_{A'B'}) \right)$. According to the (communicated) measurement results~$(\alpha,\beta)$, Alice and Bob apply local post-processing operations~$\cM^{\alpha,\beta}$ and~$\cN^{\alpha,\beta}$, respectively. The measurements and postprocessing operations are chosen  such that the resulting average state~$\hat{\rho}_{AB}=\sum_{\alpha,\beta} p(\alpha,\beta) \left(\cM^{\alpha,\beta}\otimes\cN^{\alpha,\beta}\right)(\rho^{\alpha,\beta}_{A_2B_2})$
is close to the target state $U_{AB}\rho_{AB}U_{AB}^\dagger$. 
\label{fig:simultaneouscomputation}}
\end{figure}

Explicitly, a general protocol in this model  proceeds as follows (see Fig.~\ref{fig:simultaneouscomputation}):
\begin{enumerate}
\item
Alice and Bob simultaneously apply (partial) local measurements $\{E^\alpha_{A_1}\}_\alpha$ and  $\{F^\beta_{B_1}\}_\beta$ to the joint state $\rho_{AB}\otimes\eta_{A'B'}$, where $\eta_{A'B'}$ is the shared entanglement. Here we have partitioned Alice's complete system $AA'$ into subsystems~$A_1$ and $A_2$, and similarly for Bob.
\item
 Alice and Bob then simultaneously communicate $\alpha$ and $\beta$ to each other.
\item
Subsequently, Alice applies a local postprocessing CPTP map $\cM^{\alpha,\beta}:\cB(A_2)\rightarrow \cB(A)$ chosen according to the measurement outcomes $(\alpha,\beta)$. Bob similarly applies a  postprocessing CPTP map~$\cN^{\alpha,\beta}:\cB(A_2)\rightarrow\cB(B)$. 
\end{enumerate}
Vaidman's scheme gives, for every unitary~$U_{AB}$, measurements and postprocessing operations such that the final state after these operations is close to~$U_{AB}\rho_{AB}U_{AB}^\dagger$.

Our focus here is on the amount of shared entanglement required for the implementation of such primitives.
 Assuming that~$A$ and~$B$ consist of $n$ qubits each, we give procedures which solve these tasks to arbitrary constant precision while consuming~$O(n2^{8n})$~ebits of entanglement. In contrast, earlier schemes based on Vaidman's ideas require an amount of entanglement scaling doubly exponentially with~$n$.  Our protocols are considerably simpler because they are based on  a modified version of teleportation proposed by  Ishizaka and Hiroshima~\cite{ishietal08,ishietal09}.

It is worth emphasizing  that our procedures only use the POVM~$\{O_{AB}^{\gamma}\}$ or the unitary~$U_{AB}$ as a black box (and do not depend on their particular form), and moreover are universal in the sense that the number of required ebits is independent of the particular measurement or unitary that is implemented. For specific instances, much more  efficient schemes are known: in~\cite{clarketal10},  it was shown that it is possible to realize a non-local unitary~$U$ by consuming an amount of entanglement exponential in~$M$, where~$M$ is the length of a factorization of $U=R_1\cdots R_M$ into Pauli rotations~$\{R_j\}_j$. Note, however, that for a generic $n$-qubit unitary,  $M$~is exponentially large in~$n$, and then our scheme again provides an exponential saving for a typical unitary.

 On the negative side, by constructing a specific example, we show that a linear number of ebits is generally insufficient to perform an instantaneous distributed measurement. This impossibility result improves upon~\cite{buhrmanetal10}, where it is shown that a certain such task cannot be performed without shared entanglement. 

Our results significantly tighten previously known upper and lower bounds on the amount of entanglement required for instantaneous non-local measurement and computation, but still leave an exponential gap. We conjecture that the entanglement scaling of our protocols is essentially optimal 
among protocols which only make black-box use of the unitary or measurement, as explained above.

These results have direct application to position-based quantum cryptography, which attempts to exploit the location of an entity as its only credential. Our techniques show that any such scheme is insecure in the presence of malicious players sharing an  amount of entanglement exponential
in the number of communicated qubits. This was previously known only when the amount of entanglement is doubly exponential~\cite{buhrmanetal10}.

On the other hand, we prove the security of certain protocols assuming that the adversarial players have less entanglement than a given linear bound and are restricted to classical communication. We show that  under these assumptions, certain protocols have exponential soundness, i.e., the adversarial players have a negligible probability of cheating successfully. In contrast, previously known security proofs achieving exponential soundness~\cite{buhrmanetal10} do not allow any prior entanglement (while also requiring the restriction to classical communication, see~discussion after equation~\eqref{eq:sequentialcompositionL} below). 

We present two types of protocols with exponential soundness assuming limited entanglement and classical communication: the first one relies on the impossibility of implementing a certain high-dimensional instantaneous measurement, whereas the second one is obtained by parallel composition and reduction to the case of adversaries without prior entanglement. The latter scenario was previously analyzed in~\cite{buhrmanetal10}. It gives rise to a protocol which can be realized using single-qubit manipulations only.

The restriction to (unlimited, but) classical instead of quantum  information  is  natural in the study of resource requirements for instantaneous measurement and computation. Indeed, allowing (unlimited) quantum communication would render instantaneous measurements trivial. Also, given an implementation of an instantaneous quantum computation, quantum communication can easily be traded against a corresponding increase in prior entanglement (using teleportation). However, in position-based quantum cryptography, the restriction to classical communication is less natural: ideally, security  should be established even in the case where  the adversaries are allowed to communicate an arbitrary amount of quantum information.  Such strong results are currently only known with constant soundness and a zero or constant amount of prior entanglement. Establishing similar results for e.g., a linear amount of entanglement and exponential soundness is a major open problem in this context. The following table summarizes this state of affairs:

\begin{center}
\begin{tabular}{c|c|c|c}
soundness & allowed communication & pre-shared ebits & reference\\
\hline
constant & quantum & zero/constant & \cite{buhrmanetal10,laulo10}\\
\hline
exponential & classical & zero & \cite{buhrmanetal10}\\
\hline
exponential & classical & linear & Lemma~\ref{lem:pinsecurityprotocol}/Lemma~\ref{lem:pionecomposition}\\
\hline
exponential & quantum & e.g., linear & open problem
\end{tabular}
\end{center}

\section{Review of Vaidman's scheme: teleportation without communication\label{sec:vaidman}}
We briefly discuss  the implementation of a bipartite unitary~$U_{AB}$ using Vaidman's scheme~\cite{vaidman03}. While this is not directly needed for the discussion of our scheme, it helps to clarify the nature of our simplifications. In particular,  in Section~\ref{sec:vaidmanentanglement},  we will discuss how the  doubly exponential scaling of the entanglement consumption arises in Vaidman's procedure.

Vaidman's scheme assumes that Alice and Bob share a large supply of EPR pairs $\ket{\Phi}_{A'B'}=\frac{1}{\sqrt{2}}(\ket{0}_{A'}\ket{0}_{B'}+\ket{1}_{A'}\ket{1}_{B'})$. This entanglement can in principle be used for teleportation~\cite{Bennett93}: to teleport a qubit state $\ket{\Psi}_A$ to Bob, Alice measures~$A$ and her part~$A'$ of an EPR pair in the Bell basis. She obtains each outcome $k\in\{0,1,2,3\}$ with equal probability~$1/4$. We will refer to such a measurement as a {\em teleportation measurement}. Conditioned on her outcome being~$k$, Bob's register $B'$ contains the state $\sigma_{k}\ket{\Psi}_A$, where $\{\sigma_i\}_{i=0}^3$ are the Pauli operators. Standard teleportation then proceeds by Alice sending~$k$ to Bob, and Bob applying the correction operation $\sigma_k$. 

Clearly, $n$~shared EPR pairs can be used to teleport an arbitrary $n$-qubit state $\ket{\Psi}$: Alice's teleportation measurement is a tensor product measurement between each qubit and one half of an EPR pair. It gives each outcome $\myvec{k}\in\{0,1,2,3\}^n$ with probability~$4^{-n}$. Bob's state in his part of the EPR pairs then ends up being $\sigma_{\myvec{k}}\ket{\Psi}_A$, where $\sigma_{\myvec{k}}=\sigma_{k_1}\otimes\cdots\otimes\sigma_{k_n}$. Sending $\myvec{k}\in\{0,1,2,3\}^n$ to Bob allows him to apply the Pauli correction~$\sigma_{\myvec{k}}$.

In a setting with free classical communication, applying a unitary to their joint state  would easily be achievable using teleportation: Bob  teleports~$n$ qubits (system $B$) to Alice using $n$~ebits of entanglement. Alice applies $U_{AB}$ and teleports the system~$B$ back to Bob. However, teleportation cannot directly be used for instantaneous non-local computation because outcomes of teleportation measurements cannot be communicated. To get around this problem, Vaidman uses an elaborate recursive technique which we explain in Section~\ref{sec:vaidmanrecursive}. 

\subsection{Reduction to a state held by one of the parties}
It is instructive to see that, while teleportation is not directly applicable in the setting of instantaneous measurement, it can nevertheless be used advantageously.  In the following we show that in order to implement a bipartite unitary~$U=U_{AB}$ applied to a state $\ket{\Psi}=\ket{\Psi}_{AB}\in(\mathbb{C}^2)^{\otimes n}\otimes (\mathbb{C}^2)^{\otimes n} $ it suffices to provide a protocol $\cP'$ with the following properties: it proceeds by application of local operations only (no communication is allowed) and  after completion of the protocol, Bob holds the state~$\sigma_{\myvec{s}}U\ket{\Psi}$ in one of his registers, and Alice and Bob have classical information $\alpha$ and $\beta$ (measurement outcomes) that together determine~$\myvec{s}=\myvec{s}(\alpha,\beta)\in\{0,1,2,3\}^{2n}$. If we design such a $\cP'$, the procedure of implementing the unitary $U$ on $\ket{\Psi}$ is immediate:

\begin{description}
\item[1.a] Alice and Bob run $\cP'$, getting classical outcomes $(\alpha,\beta)$ and the state $\sigma_{\myvec{s}}U\ket{\Psi}$ in Bob's registers $B'_1B'_2$, where $\myvec{s}=\myvec(\alpha,\beta)$. 
\item[1.b] Bob performs a teleportation measurement on $B'_1$ and $n$~EPR pairs shared between Alice and Bob in registers $A'':B''$. He obtains the outcome $\myvec{v}\in\{0,1,2,3\}^n$. 
\item[2] Alice sends $\alpha$ to Bob. Bob sends $\beta$ and $\myvec{v}$  to Alice.
\item[3] Alice and Bob both compute $\myvec{s}=\myvec{s}(\alpha,\beta)$. Let $\myvec{s}=(\myvec{s}_A,\myvec{s}_B)\in\{0,1,2,3\}^n\times\{0,1,2,3\}^n$. Alice applies $\sigma_{\myvec{s}_A}\sigma_{\myvec{v}}$ to $A''$, while Bob applies $\sigma_{\myvec{s}_B}$ to~$B'_2$.
\end{description}
At the end of this protocol Alice and Bob share $U\ket{\Psi}$ in $A'':B_2'$.

\subsection{Vaidman's recursive scheme\label{sec:vaidmanrecursive}}
We now show how Vaidman realizes a protocol $\cP'$ as described in the previous section. Even though the final measurements are done instantaneously and simultaneously, it is useful for the construction to think of an interactive procedure. Since no communication is allowed in $\cP'$ and Alice and Bob's measurements commute, this interactive protocol is indeed instantaneous.

 In the first round, Bob performs a teleportation measurement on~$B$ and his part of $n$~ebits in registers~$A'_1:B'_1$. Conditioned on the outcomes being~$\myvec{t}_1\in\{0,1,2,3\}^n$, Alice now has the state $(\id_A\otimes \sigma_{\myvec{t}_1})\ket{\Psi}$ in $AA'_1$. Since Alice is ignorant of~$\myvec{t}_1$, she cannot apply the corresponding correction operation. Instead, she simply applies $U^{(1)}=U$ to her state, and then performs a teleportation measurement between the resulting state and $2n$~ebits in registers~$A_1'':B_1''$ shared between Alice and Bob. Denoting by $\myvec{s}_1\in\{0,1,2,3\}^{2n}$ be the outcome of Alice's measurement, Bob holds the state~$\sigma_{\myvec{s}_1}U(\id_A\otimes \sigma_{\myvec{t}_1})\ket{\Psi}$ in his register~$B_1''$, which can be written as
\begin{equation}
U^{(2)}_{\myvec{s}_1}(\myvec{t}_1)^\dagger U\ket{\Psi}\qquad\textrm{ where }\qquad U^{(2)}_{\myvec{s}_1}(\myvec{t}_1)=U(\id_A\otimes\sigma_{\myvec{t}_1})U^\dagger \sigma_{\myvec{s}_1}\ .\label{eq:firstround}
\end{equation}
Clearly, if $\myvec{t}_1=0^n$, then $U^{(2)}_{\myvec{s}_1}(\myvec{t}_1)=\sigma_{\myvec{s}_1}$. This means that Bob has a state of the desired form and can stop. In the actual protocol, this means that Bob does not perform further measurements.  However, $\myvec{t}_1=0^n$ happens only with probability~$4^{-n}$. 

Vaidman's crucial insight was that it is possible to recursively apply this procedure, essentially attempting to implement $U^{(2)}_{\myvec{s}_1}(\myvec{t}_1)$ in the next round. However, this is not entirely straightforward since Bob's measurement outcome~$\myvec{t}_1$ (unlike $\myvec{s}_1$) is unknown to Alice. To get around this, Alice and Bob use, for every possible outcome~$\myvec{\hat{t}}_1$ of Bob's measurement result,   a separate set of~$2n$~ebits in registers~$A'_{1,\myvec{\hat{t}}_1}:B'_{1,\myvec{\hat{t}}_1}$ for the Bob's teleportation measurements, and $2n$~ebits in registers~$A''_{1,\myvec{\hat{t}}_1}:B''_{1,\myvec{\hat{t}}_1}$ for Alice's measurements.  In essence, this allows Alice to implement operations which effectively depend on the outcomes~$\myvec{t}_1$ of the previous round. She just applies, for each~$\myvec{\hat{t}}_1$, a suitably chosen operation to registers $A'_{1,\myvec{\hat{t}}_1}$ and $A''_{1,\myvec{\hat{t}}_1}$ which may depend on $\myvec{\hat{t}}_1$. 
Because Bob holds~$\myvec{t}_1$, he knows which 
pair~$(B'_{1,\myvec{\hat{t}}_1},B''_{1,\myvec{\hat{t}}_1})$ of registers is the relevant one containing the desired state.

Explicitly, in the second round, Bob performs a teleportation measurement on $B''_1$ and his part of the $2n$~ebits~$A'_{1,\myvec{t}_1}:B'_{1,\myvec{t}_1}$. Let $\myvec{t}_2\in\{0,1,2,3\}^{2n}$ be the corresponding outcomes. He does not use the other registers.  Then Alice, for each $\myvec{\hat{t}}_1\in\{0,1,2,3\}^{n}$, applies $U^{(2)}_{\myvec{s}_1}(\myvec{\hat{t}}_1)$ to register~$A'_{1,\myvec{\hat{t}}_1}$ and performs a teleportation measurement between $A'_{1,\myvec{\hat{t}}_1}$ and~$A''_{1,\myvec{\hat{t}}_1}$ getting outcome $\myvec{s}_2(\myvec{\hat{t}}_1)$. At the end of these operations, Bob, in register~$B_{1,\myvec{t}_1}''$, holds
\begin{equation}
\hspace{-6ex}U^{(3)}_{\myvec{s}_1,\myvec{s}_2}(\myvec{t}_1,\myvec{t}_2)^\dagger U\ket{\Psi}\quad\textrm{ where }\quad
U^{(3)}_{\myvec{s}_1,\myvec{s}_2}(\myvec{t}_1,\myvec{t}_2)=U^{(2)}_{\myvec{s}_1}(\myvec{t}_1)\sigma_{\myvec{t}_2}U^{(2)}_{\myvec{s}_1}(\myvec{t}_1)^\dagger \sigma_{\myvec{s}_2}\ \label{eq:firstsecondround}
\end{equation}
where $\myvec{s}_2=\myvec{s}_2(\myvec{t}_1)$, and 
he neglects other registers $B''_{1,\myvec{\hat{t}}_1}$, $\myvec{\hat{t}}_1\neq\myvec{t}_1$. Similarly as before, if $\myvec{t}_2=0^{2n}$, which happens with probability~$4^{-2n}$, Bob has $\sigma_{\myvec{s}_2}U\ket{\Psi}$ and has reached the goal.  Note that the total number of ebits used in the second round is $4n\cdot 4^n$. 

It is clear how to continue this recursion: in the $R$-th round, Alice and Bob use, for every possible sequence $(\myvec{\hat{t}}_1,\ldots,\myvec{\hat{t}}_{R-1})$ of Bob's outcomes in the previous rounds, $2n$~ebits $A'_{1,\myvec{\hat{t}}_1,\ldots,\myvec{\hat{t}}_{R-1}}:B'_{1,\myvec{\hat{t}}_1,\ldots,\myvec{\hat{t}}_{R-1}}$ and $2n$~ebits $A''_{1,\myvec{\hat{t}}_1,\ldots,\myvec{\hat{t}}_{R-1}}:B''_{1,\myvec{\hat{t}}_1,\ldots,\myvec{\hat{t}}_{R-1}}$.  Bob performs a teleportation measurement between $B''_{1,\myvec{t}_1,\ldots,\myvec{t}_{R-2}}$ and his part of
the $2n$~ebits $A'_{1,\myvec{t}_1,\ldots,\myvec{t}_{R-1}}:B'_{1,\myvec{t}_1,\ldots,\myvec{t}_{R-1}}$. For each sequence~$(\myvec{\hat{t}}_1,\ldots,\myvec{\hat{t}}_{R-1})$, Alice applies the unitary $U^{(R)}_{\myvec{\hat{s}}_1,\myvec{\hat{s}}_2,\ldots,\myvec{\hat{s}}_{R-1}}(\myvec{\hat{t}}_1,\ldots,\myvec{\hat{t}}_{R-1})$ to $A'_{1,\myvec{\hat{t}}_1,\ldots,\myvec{\hat{t}}_{R-1}}$ where $\myvec{\hat{s}}_j=\myvec{\hat{s}}_j(
\myvec{\hat{t}}_1,\ldots,\myvec{\hat{t}}_{j-1})$. She also performs a teleportation measurement between $A'_{1,\myvec{\hat{t}}_1,\ldots,\myvec{\hat{t}}_{R-1}}$ and her part of the $2n$~ebits $A''_{1,\myvec{\hat{t}}_1,\ldots,\myvec{\hat{t}}_{R-1}}:B''_{1,\myvec{\hat{t}}_1,\ldots,\myvec{\hat{t}}_{R-1}}$. Bob, in register~$B_{1,\myvec{t}_1,\ldots,\myvec{t}_{R-1}}''$
ends up with 
\begin{eqnarray}
\hspace{-6ex}U^{(R+1)}_{\myvec{s}_1,\ldots,\myvec{s}_R}(\myvec{t}_1,\ldots,\myvec{t}_R)^\dagger U\ket{\Psi}\quad\textrm{ where }\nonumber\\
\hspace{-14ex} U^{(R+1)}_{\myvec{s}_1,\ldots,\myvec{s}_R}(\myvec{t}_1,\ldots,\myvec{t}_R)=U^{(R)}_{\myvec{s}_1,\ldots,\myvec{s}_{R-1}}(\myvec{t}_1,\ldots,\myvec{t}_{R-1})\sigma_{\myvec{t}_R}U^{(R)}_{\myvec{s}_1,\ldots,\myvec{s}_{R-1}}(\myvec{t}_1,\ldots,\myvec{t}_{R-1})^\dagger \sigma_{\myvec{s}_{R}}\ .\label{eq:Rthround}
\end{eqnarray} 
Each round leads to a trivial correction operation $\myvec{t}_R=0^{2n}$ with probability~$4^{-2n}$. As soon as Bob obtains this outcome in some round~$R$, he stops performing any further measurements. Alice, on the other hand, continues applying her operations until she has operated on all the entangled states available to the two parties.

The classical information $\alpha$ sent by Alice consists of the sequence of all her measurement results $\{\myvec{s}_j(\myvec{t})\}_{j,\myvec{t}}$. Bob's message $\beta$ consist of the number of rounds~$R$ and all his measurement results $(\myvec{t}_1,\ldots,\myvec{t}_R)$. Bob holds $\sigma_{\myvec{s}_R}U\ket{\Psi}$  in $B''_{1,\myvec{t}_1,\ldots,\myvec{t}_{R-1}}$, where $\myvec{s}_R$ is determined by $(\alpha,\beta)$, as required.

\subsection{Entanglement consumption in Vaidman's scheme\label{sec:vaidmanentanglement}}
The amount of entanglement required in~$R$ rounds is roughly proportional to the number of sequences $(\myvec{\hat{t}}_1,\ldots,\myvec{\hat{t}}_{R-1})$ and therefore exponential in~$R$.  Since the probability of success in each round (except the first round which is $4^{-n}$) equals $4^{-2n}$, to reach the success probability $1-\varepsilon$, $R$ needs to be roughly $\log(1/\varepsilon)\cdot2^{4n}$. As a result, the amount of required entanglement for a constant $\varepsilon$ is doubly exponential in~$n$.

Clearly, the reason for this unfavorable behavior is the recursive structure of the protocol, which is a consequence of the non-interactive way the teleportation correction operations are dealt with. Our simplified schemes avoid these problems altogether by making use of a different kind of teleportation scheme whose correction operations are in some sense trivial. In particular, the resulting scheme is non-recursive.

\section{Port-based teleportation}
In this section, we review a form of teleportation introduced by Ishizaka and Hiroshima~\cite{ishietal08,ishietal09}.  To distinguish this kind of teleportation from the better known usual scheme, we borrow from their terminology and call this {\em port-based teleportation}.

\subsection{Teleportation without correction}
 
The goal of port-based teleportation is to achieve teleportation with simpler correction operations on Bob's side. Instead of being able to apply arbitrary Pauli operations, we assume that Bob can only perform the arguably simplest imaginable CPTP map depending on classical information. That is, Bob can discard any subsystem of his choosing according to the classical information received from Alice.

Remarkably, teleportation is still possible in this restricted setting using a more intricate measurement on Alice's side.  Concretely, Alice wants to teleport a qudit state $\ket{\Psi}_A$ from her system~$A\cong\mathbb{C}^d$ to Bob's system~$B\cong\mathbb{C}^d$. We assume that Alice and Bob share~$N$  copies of the maximally entangled state $\ket{\Phi}=\frac{1}{\sqrt{d}}\sum_{i=1}^d \ket{i}\ket{i}$ in registers $A'_1:B'_1,A'_2:B'_2,\ldots,A'_N:B'_N$. We fix an orthonormal standard basis in each of these spaces. The protocol proposed in~\cite{ishietal08} then proceeds as follows:
\begin{enumerate}
\item
Alice performs a certain POVM $\{E^i_{AA'^N}\}_{i=1}^N$ on her systems, where $A'^N=A'_1\cdots A'_N$. She sends the result~$i$ to Bob.
\item
Bob discards everything except the subsystem $B'_i$ and calls it $B$. This register is supposed to hold the state $\ket{\Psi}$.   
\end{enumerate}
Let $\cE_{\ket{\Phi}^{\otimes N}}:\cB(A)\rightarrow\cB(B)$ be the CPTP map described by this protocol taking the original state $\ket{\Psi}_A$ to Bob's system~$B_i'\cong B$ by using auxiliary entanglement $\ket{\Phi}^{\otimes N}$. (Here~$\cB(A)$ denotes the set of bounded linear operators on a Hilbert space $A$.) Explicitly, this map  is given by
\begin{equation}
\hspace{-6ex}\cE_{\ket{\Phi}^{\otimes N}}(\rho_A)=\sum_{i=1}^N\tr_{B'^N\backslash B'_i} \tr_{AA'^N}\left((\id_{B'^N}\otimes E^{i}_{AA'^N})(\rho_A\otimes\proj{\Phi}^{\otimes N}_{A'^NB'^N})\right)\ .\label{eq:multioutputchannel}
\end{equation}
In this expression,~$\tr_{B'^N\backslash B'_i}:\cB(B'^N)\rightarrow\cB(B)$ denotes the CPTP map consisting of tracing out all systems except~$B'_i\cong B$. The main result of~\cite{ishietal08,ishietal09} is that there exists a choice of a POVM~$\{E^i_{AA'^N}\}_i$ such that 
this CPTP map is close to the identity channel for large~$N$. 

The measure of distance used in~\cite{ishietal08,ishietal09} to compare a channel $\cE:\cB(A)\rightarrow\cB(B)$  to the identity channel $\cI_A:\cB(A)\rightarrow \cB(A) \cong\cB(B)$ is the  entanglement fidelity 
\begin{equation}
F(\cE) = \tr \proj{\Phi}_{BC}(\cE\otimes\cI_C)(\proj{\Phi}_{AC})\ ,\label{eq:entanglementfidelity}
\end{equation}
where $\ket{\Phi}$ is the maximally entangled state. Because of the relation~\cite{horodeckietal99}
\begin{equation*}
(F(\cE)d+1)/(d+1)=\int \bra{\Psi}\cE(\proj{\Psi})\ket{\Psi}\ d\Psi\ ,
\end{equation*} between entanglement fidelity $F(\cE)$ and the output fidelity averaged over the Haar measure on all pure input states, a lower bound on~\eqref{eq:entanglementfidelity}  expresses how well~$\cE$ preserves  quantum information for random inputs on average. The following is shown in~\cite{ishietal08,ishietal09}. 
\begin{theorem}[\cite{ishietal08,ishietal09}]\label{thm:portbasedmain}
Let $\ket{\Phi}=\frac{1}{\sqrt{d}}\sum_{i=1}^d \ket{i}\ket{i}$ be the maximally entangled state.  There is a POVM
$\{E^i_{AA'^N}\}_{i=1}^N$ on $\mathbb{C}^d\otimes(\mathbb{C}^d)^{\otimes N}$ such that the CPTP map~$\cE_{\ket{\Phi}^{\otimes N}}$ defined by~\eqref{eq:multioutputchannel} satisfies
\begin{equation*}
F(\cE_{\ket{\Phi}^{\otimes N}})\geq 1-\frac{d^2-1}{N}\ .
\end{equation*}
\end{theorem}
\noindent Due to the importance of port-based teleportation to our schemes for instantaneous computation, we give a detailed derivation of Theorem~\ref{thm:portbasedmain} in \ref{app:portbased}.

Instead of measuring average-case closeness to the identity channel, it is desirable to have worst-case bounds. In other words, we would like to show that the state on~$B$ at the end of the port-based teleportation scheme is close to the original input state on~$A$ for all inputs.  This requires using a distance measure different from the entanglement fidelity~\eqref{eq:entanglementfidelity}. A natural distance measure on the set of CPTP maps is the completely bounded trace norm or diamond norm denoted $\|\cdot \|_\diamond$. This norm is defined in terms of the trace norm 
\begin{equation*}
\|L\|_1 = \tr\sqrt{L^\dagger L}\ ,\qquad L\in \cB(A)\ .
\end{equation*}  The trace norm induces a norm
\begin{equation*}
\|\Omega\|_1 =\max_{L\in\cB(A): \|L\|_1\leq 1} \|\Omega(L)\|_1
\end{equation*}
on the set of superoperators $\Omega:\cB(A)\rightarrow\cB(B)$. The diamond norm is defined as
\begin{equation}
\|\Omega\|_\diamond =\sup_{k\geq 1}\|\Omega\otimes\cI_{\mathbb{C}^k}\|_1\ \label{eq:diamondnorm}
\end{equation}
where $\cI_{\mathbb{C}^k}$ is the identity (super)operator on $\cB(\mathbb{C}^k)$. 

Given two CPTP maps $\cE:\cB(A)\rightarrow\cB(B)$ and  $\cF:\cB(A)\rightarrow\cB(B)$,~$\|\cE-\cF\|_{\diamond}$ is a natural measure of their distance because of the following operational interpretation. The quantity
\begin{equation*}
\frac{1}{2}+\frac{1}{4}\|\cE-\cF\|_{\diamond}
\end{equation*}
is equal to the probability of successfully distinguishing~$\cE$ from~$\cF$ when a single instance of either one of these channels is provided with equal prior probability. This implies that the diamond-distance can only decrease under composition of maps. 

\begin{corollary}\label{corollary:diamon}
Let $\ket{\Phi}=\frac{1}{\sqrt{d}}\sum_{i=1}^d \ket{i}\ket{i}$ be the maximally entangled state.  There is a POVM
$\{E^i_{AA'^N}\}_{i=1}^N$ on $\mathbb{C}^d\otimes(\mathbb{C}^d)^{\otimes N}$ such that the CPTP map~$\cE_{\ket{\Phi}^{\otimes N}}$ defined by~\eqref{eq:multioutputchannel} satisfies
\begin{equation*}
\|\cE_{\ket{\Phi}^{\otimes N}}-\cI_{\mathbb{C}^d}\|_\diamond\leq \frac{4d^2}{\sqrt{N}} \ .
\end{equation*}
\end{corollary}
\begin{proof}
Consider a superoperator $\Omega:\cB(A)\rightarrow\cB(B)$ and let $C\cong A$. The Choi-Jamiolkowski representation~\cite{choi,jamiolkowski} of~$\Omega$ is the operator
\begin{equation*}
J(\Omega):=(\Omega\otimes\cI_C)(\proj{\Phi}_{AC})\ .
\end{equation*}
$J(\Omega)$ is a quantum state if $\Omega$ is a CPTP map.
Since $J(\cI_{\mathbb{C}^d})=\proj{\Phi}$, the entanglement fidelity $F(\cE)$ is equal to the overlap  of the states $J(\cI_{\mathbb{C}^d})$ and $J(\cE)$, i.e.,
\begin{equation*}
F(\cE)=\tr(J(\cI_{\mathbb{C}^d})J(\cE))\ . 
\end{equation*}
Because $J(\cI_{\mathbb{C}^d})$ is pure, we can use the general inequality (see e.g.,~\cite{fuchsgraaf97})
\begin{equation*}
\frac{1}{2}\big\|\rho-\proj{\Psi}\big\|_1\leq  \sqrt{1-\bra{\Psi}\rho\ket{\Psi}}
\end{equation*}
bounding the distance between a pure state $\ket{\Psi}$ and a mixed state $\rho$. We conclude that 
\begin{equation}
\frac{1}{2}\|J(\cE_{\ket{\Phi}^{\otimes N}})-J(\cI_{\mathbb{C}^d})\|_1\leq \sqrt{1-F(\cE)}\ .\label{eq:choibound}
\end{equation}
The claim then follows from the linearity of the map $J$, inequality~\eqref{eq:choibound} and the following lemma applied to~$\Omega=\cE_{\ket{\Phi}^{\otimes N}}-\cI_{\mathbb{C}^d}$.
\end{proof}

\begin{lemma} For every two CPTP maps $\Omega_1, \Omega_2:\cB(A)\rightarrow\cB(B)$ we have
\begin{equation*}
\|\Omega_1 - \Omega_2\|_\diamond \leq 2 (\dim A) \|J(\Omega_1) - J(\Omega_2)\|_1\ .
\end{equation*}
\end{lemma}
\begin{proof}
Let $\Omega= \Omega_1 - \Omega_2$ and $C\cong A$. It is well-known (see e.g.,~\cite{watrous09}) that for any such~$\Omega$ 
\begin{eqnarray*}
\|\Omega\|_\diamond =  \| \Omega\otimes \cI_{C} \|_{1}  &=   \max_{\ket{\Psi}_{AC}}  \| (\Omega\otimes \cI_C) (\proj{\Psi})   \|_{1}\\
&=2\max_{\ket{\Psi}_{AC}, 0\leq P_{BC}\leq \id_{BC}}\tr\left(P_{BC} (\Omega\otimes \cI_C) (\proj{\Psi}) \right)\ .
\end{eqnarray*}
On the other hand, for every $\ket{\Psi}_{AC}$ there exists $M_C$ such that $\ket{\Psi}_{AC} = \id_A \otimes M_C\ket{\Phi}_{AC}$, where due to normalization $\tr M_C^{\dagger} M_C = \dim A$. Then we have
\begin{equation*}
\Omega \otimes \cI_C ( \proj{\Psi}) = \id_B\otimes M_C\, J(\Omega)_{BC} \,\id_B \otimes M_C^{\dagger}\ .
\end{equation*}
As a result,
\begin{eqnarray*}
\|\Omega\|_\diamond  &=2\max_{
M_{C}, 0\leq P_{BC}\leq \id_{BC}}\tr\left(   (\id_B\otimes M^{\dagger}_C\, P_{BC}\,  \id_B\otimes M_C)  J(\Omega)_{BC}\right)\\
& \leq 2 \max_{
M_{C}, 
0\leq P_{BC}\leq \id_{BC}}\| \id_B\otimes M^{\dagger}_C\, P_{BC}\,  \id_B\otimes M_C\|_{\infty} \cdot \|J(\Omega)_{BC}\|_1 \\
& \leq 2(\dim A) \|J(\Omega)_{BC}\|_1\ ,
\end{eqnarray*}
where $\|\cdot \|_{\infty}$ denotes the operator norm, and in the last line we use the normalization of $M_C$.
\end{proof}

\section{Protocols for instantaneous measurement and computation}
Here we propose and analyze two novel protocols for instantaneous measurement and computation. Both protocols depend on a parameter~$N$ which captures the amount of entanglement consumed and the accuracy achieved by the protocol. The first protocol $\pi_N(\cO)$  gives an instantaneous realization of a (non-local) POVM~$\cO=\{O^\gamma_{AB}\}_{\gamma}$.  The second protocol~$\pi_N(U)$ implements a non-local unitary~$U_{AB}$. In the following descriptions, we  divide up the instantaneous (simultaneous) application of measurements by both Alice and Bob into several stages to simplify the analysis. Note, however, that the actions of Alice and Bob commute and do not have to be performed in the prescribed order.

\begin{protocol}{$\pi_N(\cO)$}{implementation of POVM $\cO=\{O^\gamma_{AB}\}_\gamma$ on a state $\ket{\Psi}_{AB}\in(\mathbb{C}^2)^{\otimes n}\otimes (\mathbb{C}^2)^{\otimes n}$. }{\newline
Alice and Bob share $n$~ebits of auxiliary entanglement in $A':B'$, and for every $j\in\{1,\ldots,N\}$, $2n$~ebits of entanglement in $A_j'':B_j''$. We write~$A''^N= A''_1\cdots A''_N$ and $B''^N=B''_1\cdots B''_N$.}\label{proto:POVM}

\item[1(a).] Bob performs a teleportation measurement between $B$ and $B'$ with outcomes $\myvec{t}\in\{0,1,2,3\}^n$. As a result, Alice holds the bipartite state $(\id_A\otimes \sigma_{\myvec{t}})\ket{\Psi}_{AA'}$.

\item[1(b).]
Alice applies the port-based teleportation-measurement on 
her $2n$~qubits in systems $AA'$ and her part of the shared entanglement in $A''^N:B''^N$. She gets an index~$i\in\{1,\ldots,N\}$, and  Bob obtains $(\id_A\otimes\sigma_{\myvec{t}})\ket{\Psi}$ (with high fidelity) in the $i$-th system $B''_i$.
\item[1(c).]
For each $j\in\{1,\ldots,N\}$, Bob first applies $\id\otimes\sigma_{\myvec{t}}$ to~$B''_j$ and then measures it using the POVM $\cO$. Let $\gamma_j$ be the outcome of this measurement. 
\item[2.]
Alice sends $i$, and Bob sends the list $\{(j,\gamma_j)\}_j$ to Charlie. 
\item[3.]
Upon receiving this classical information, Charlie outputs $\gamma_i$.
\end{protocol}

\begin{protocol}{$\pi_{N}(U)$}{implementation of a unitary~$U_{AB}$ on a state $\ket{\Psi}_{AB}\in (\mathbb{C}^2)^{\otimes n}\otimes (\mathbb{C}^2)^{\otimes n}$}{
\newline
Alice and Bob share auxiliary systems $A'A''^N:B'B''^N$ as in protocol~$\pi_N(\cO)$. Here we assume that each $A''_j$ is partitioned into an $A$-part and a $B$-part (with $n$~qubits each), and similarly for~$B''_j$. They additionally share, for every $j\in\{1,\ldots,N\}$, $n$~ebits of entanglement in systems $A'''_j:B'''_j$. 
}\label{proto:unitary}
\item[1(a)-1(b).]  Execute steps~1(a) and 1(b) of protocol~$\pi_N(\cO)$. 
\item[1(c).] For every $j\in\{1,\ldots,N\}$, Bob applies $U(\id\otimes \sigma_{\myvec{t}})$ to $B''_j$. Then he performs a (usual) teleportation measurement between the $A$-part of $B''_j$ and $B'''_j$. Letting ${\myvec{v}}_j\in\{0,1,2,3\}^{n}$ be the outcome of this measurement, systems $A'''_i$ and the $B$-part of $B''_i$ now contain $(\sigma_{\myvec{v}_i}\otimes\id)U \ket{\Psi}$ (with high fidelity). 
\item[2.] Alice sends $i$ to Bob, and Bob sends the list $\{(j,\myvec{v}_j)\}_j$ to Alice.
\item[3.]
Alice discards everything except $A'''_i$, on which she applies $\sigma_{\myvec{v}_i}$. Bob discards everything except the $B$-part of $B''_i$.
\end{protocol}

To quantitatively express the accuracy of these protocols, we use the diamond norm. For this to make sense for POVMs,  we regard a POVM $\cE=\{E_i\}_{i}$ as a CPTP map with output diagonal in the standard basis, i.e., $\cE(\rho)=\sum_i \tr(E_i\rho)\proj{i}$. The following theorem is an easy consequence of Corollary~\ref{corollary:diamon} and the fact that the diamond-distance does not increase under the composition of CPTP maps. 
\begin{theorem}\label{thm:maintheorem}
The protocols introduced in this section have the following properties for any $\varepsilon>0$. 
\begin{enumerate}[(i)]
\item\label{it:instmeasurement}
Let $\cO=\{O^\gamma_{AB}\}_\gamma$ be a bipartite POVM on $(\mathbb{C}^2)^{\otimes n}\otimes(\mathbb{C}^2)^{\otimes n}$. Set $N:= 2^{8n+4}/\varepsilon^2$ and let $\cM$ be the POVM defined by protocol~$\pi_N(\cO)$. Then~$\cM$ approximates~$\cO$ up to accuracy
\begin{equation*}
\|\cM-\cO\|_\diamond \leq \varepsilon 
\end{equation*}
and consumes
\begin{equation*}
n\left(1+\frac{2^{8n+5}}{\varepsilon^2}\right)
\end{equation*}
ebits of entanglement. 
\item
Let $U=U_{AB}$  be a bipartite POVM on $(\mathbb{C}^2)^{\otimes n}\otimes(\mathbb{C}^2)^{\otimes n}$. Set $N:=2^{8n+4}/\varepsilon^2$ and let $\cE$ be the CPTP map defined by protocol~$\pi_N(U)$. Then~$\cE$ approximates~$U$ up to accuracy
\begin{equation*}
\|\cE-U\|_\diamond \leq \varepsilon 
\end{equation*}
while consuming
\begin{equation*}
n\left(1+\frac{3\cdot 2^{8n+4}}{\varepsilon^2}\right)\ .
\end{equation*}
ebits of entanglement. 
\end{enumerate}
\end{theorem}
These protocols and their analysis can clearly be extended in a straightforward manner to multipartite (i.e., more than bipartite)  non-local POVMs and unitaries.

\section{A lower bound\label{sec:lowerbound}}
In this section, we show that there is a measurement on~$2n$ qubits which is not realizable instantaneously with fewer than $n/2$~ebits of entanglement. Our construction is based on mutually unbiased bases. A pair of orthonormal bases~$\{\ket{e^1_x}\}_{x=1}^d$ and~$\{\ket{e^2_y}\}_{y=1}^d$ of~$\mathbb{C}^d$ is called mutually unbiased if
\begin{equation*}
|\langle e^1_x|e^2_y\rangle|^2=\frac{1}{d}\qquad\textrm{ for all }\qquad x,y\in\{1,\ldots,d\}\ .
\end{equation*}
It is known~\cite{ivanovic81,woottersfields89,bandy02} that if $d=p^n$ is a power of a prime number, then a set of $d+1$~pairwise mutually unbiased bases~$\{\cB_a:=\{\ket{e^a_x}\}_{x=1}^d\}_{a=0}^d$ exists in $\mathbb{C}^d$. We will assume that $d$ is of this form (specifically~$d=2^n$).

\begin{theorem}\label{thm:guessingprobbound}
Suppose that two parties Alice and Bob share 
one of the states 
\begin{equation*}
\rho_{AB}^x = \frac{1}{d+1} \sum_{a=0}^d \proj{a}_A\otimes \proj{e^a_x}_B\ \qquad x\in\{1,\ldots,d\}\ ,
\end{equation*}
each with prior probability $1/d$, and additionally have an (arbitrary) shared entangled state~$\eta_{A'B'}$. They are asked to output $x$ with an instantaneous measurement (cf.~Fig.~\ref{fig:instmeasurement}).   Then their success probability  is upper bounded by
\begin{equation}
p_{succ} \leq \frac{2 \dim B'}{\sqrt{d}}.\label{eq:predictionprob}
\end{equation}
\end{theorem}
\begin{proof}
Because of the linearity of the success probability as a function of the input ensemble, it suffices to prove the bound~\eqref{eq:predictionprob} in the case where Alice and Bob receive $\ket{a}$ and $\ket{e^a_x}$ respectively, each with probability $1/d(d+1)$.

Since Alice's input~$A$ is a classical register (containing $a$),
 we may assume without loss of generality that her message~$\alpha$ to Charlie is determined by measuring the register~$A'$ using a POVM $\{E_{A'}^{a, \alpha}\}_{\alpha}$ which depends on $a$. Letting $\{F_{BB'}^{\beta}\}_{\beta}$ be Bob's measurement and $g$ be the classical post-processing function, the success probability is equal to 
\begin{eqnarray*}
p_{\text{succ}} &=\frac{1}{d(d+1)} \sum_{a, x} \,\sum_{\alpha, \beta: g(\alpha, \beta) = x}  \tr \left(  (E_{A'}^{a, \alpha} \otimes F_{BB'}^{\beta}) (\proj{e^a_x}_B \otimes \eta_{A'B'} )   \right)\\
& = \frac{1}{d(d+1)} \sum_{a, x} \, \sum_{\alpha, \beta: g(\alpha, \beta)= x}  \tr \left(   F_{BB'}^{\beta}   (  \proj{e^a_x}_B  \otimes \tau_{B'}^{a, \alpha} )    \right)     \ ,
\end{eqnarray*}
where 
\begin{equation*}
\tau_{B'}^{a, \alpha} = \tr_{A'} \left(   (E_{A'}^{a, \alpha} \otimes \id_B')  \eta_{A'B'}  \right)\ .
\end{equation*}
Using 
\begin{equation}
\tau_{B'}^{a, \alpha} \leq \tr(\tau_{B'}^{a, \alpha}) \id_{B'}\ ,\label{eq:dimensionboundentangled}
\end{equation} we have
\begin{eqnarray*}
p_{\text{succ}} &\leq \frac{1}{d(d+1)} \sum_{\beta} \, \tr \left(  F_{BB'}^{\beta}   \left( \sum_{a, \alpha}   \tr(\tau_{B'}^{a, \alpha})  \proj{e^a_{g(\alpha, \beta)}}_B  \otimes \id_{B'} \right)    \right)   \\
& \leq  \frac{1}{d(d+1)} \sum_{\beta} \, \tr (  F_{BB'}^{\beta}  )  \left\| \sum_{a, \alpha}   \tr(\tau_{B'}^{a, \alpha})  \proj{e^a_{g(\alpha, \beta)}}_B  \otimes \id_{B'} \right\|_{\infty}     \\
& =  \frac{1}{d(d+1)} \sum_{\beta} \, \tr (  F_{BB'}^{\beta}  )  \left\| \sum_{a, \alpha}   \tr(\tau_{B'}^{a, \alpha})  \proj{e^a_{g(\alpha, \beta)}}_B  \right\|_{\infty} \ .
\end{eqnarray*}
If we show that 
\begin{equation}\label{eq:normineq}
 \left\| \sum_{a, \alpha}   \tr(\tau_{B'}^{a, \alpha})  \proj{e^a_{g(\alpha, \beta)}}_B  \right\|_{\infty} \leq \frac{d+2}{\sqrt{d}}+1\qquad\textrm{ for every }\beta\ ,
\end{equation}
the claim of the theorem follows because
\begin{eqnarray*}
p_{\text{succ}} & \leq   \frac{d+\sqrt{d}+2}{d(d+1)\sqrt{d}} \,\sum_{\beta} \, \tr (  F_{BB'}^{\beta}  )  =  \frac{d+\sqrt{d}+2}{d(d+1)\sqrt{d}} \tr \id_{BB'} \\
                           & =      \frac{d+\sqrt{d}+2}{(d+1)\sqrt{d}}\dim B'                                                          \leq  \frac{2\dim B'}{\sqrt{d}}\ .
\end{eqnarray*}
It remains to prove~\eqref{eq:normineq} for a fixed~$\beta$. Consider the (not necessarily normalized) vector 
\begin{equation*} 
\ket{V}_{BR} = \sum_{a, \alpha}  \sqrt{\tr( \tau_{B'}^{a, \alpha} )}  \ket{e^a_{g(\alpha, \beta)}}_B\otimes \ket{a, \alpha}_R\ ,
\end{equation*}
where $R$ is an auxiliary Hilbert space with orthonormal basis $\{\ket{a,\alpha}_R\}_{a,\alpha}$.
Denoting by~$M$ the matrix of interest in~\eqref{eq:normineq}, we have 
\begin{eqnarray*}
\| M\|_{\infty} & =\left\| \tr_{R} \proj{V}_{BR} \right\|_{\infty} =\left\|  \tr_B \proj{V}_{BR}  \right\|_{\infty}   \\
& = \left\|   \sum_{a, a', \alpha, \alpha'}  \sqrt{\tr( \tau_{B'}^{a, \alpha} )\tr( \tau_{B'}^{a', \alpha'}) } \, \langle e^{a'}_{g(\alpha', \beta)}  \vert e^a_{g(\alpha, \beta)} \rangle \vert a, \alpha\rangle \langle a', \alpha'\vert  \right\|_{\infty}\ .
\end{eqnarray*}
On the other hand, by a simple triangle inequality we obtain that for every matrix $(c_{ij})_{i, j}$, $\| (c_{ij})_{i, j}  \|_{\infty}  \leq \|  (\vert  c_{ij} \vert)_{i,j}  \|_{\infty} $, where $\vert \cdot \vert$ denotes the absolute value of a complex number. Therefore, 
\begin{equation*}
\| M \|_{\infty} \leq \mu + \nu
\end{equation*}
where  
\begin{equation*}
 \mu =   \left\|       \frac{1}{\sqrt{d}} \sum_{a, a', \alpha, \alpha'}   \sqrt{\tr( \tau_{B'}^{a, \alpha} )\tr( \tau_{B'}^{a', \alpha'}) }   \vert a, \alpha\rangle \langle a', \alpha'\vert                           \right\|_{\infty},
 \end{equation*}
 and
 \begin{equation*}
\nu= 
 \left\|   \sum_a    \sum_{\alpha, \alpha'}      \sqrt{\tr( \tau_{B'}^{a, \alpha} )\tr( \tau_{B'}^{a, \alpha'}) }  (\delta_{g(\alpha, \beta), g(\alpha', \beta)   }     - \frac{1}{\sqrt{d}} )      \vert a, \alpha\rangle \langle a, \alpha'\vert      \right\|_{\infty}\ .
\end{equation*}
The first summand can be bounded as
\begin{equation*}
\mu \leq \frac{1}{\sqrt{d}} \sum_{a, \alpha} \tr(\tau_{B'}^{a, \alpha})= \frac{1}{\sqrt{d}} \sum_a \tr (\tr_{A'} \eta_{A'B'})=\frac{d+1}{\sqrt{d}}\ .\label{eq:abound}
\end{equation*}
To bound the second term, we use the direct sum property~$\|C\oplus D\|_\infty=\max\{\|C\|_\infty,\|D\|_\infty\}$, the triangle inequality and the easily verified fact that
\begin{equation*}
\|K\|_\infty\leq \|K+L\|_\infty 
\end{equation*}
for  real symmetric matrices~$K$ and~$L$ with nonnegative entries. This gives
\begin{eqnarray} 
\nu&=&  \max_a    \left\|      \sum_{\alpha, \alpha'}      \sqrt{\tr( \tau_{B'}^{a, \alpha} )\tr( \tau_{B'}^{a, \alpha'}) } ( \delta_{g(\alpha, \beta), g(\alpha', \beta)   }     - \frac{1}{\sqrt{d}} )      \vert a, \alpha\rangle \langle a, \alpha'\vert      \right\|_{\infty}\nonumber\\
&\leq& \max_a  \left\|      \sum_{\alpha, \alpha'}      \sqrt{\tr( \tau_{B'}^{a, \alpha} )\tr( \tau_{B'}^{a, \alpha'}) }  \delta_{g(\alpha, \beta), g(\alpha', \beta)   }\ket{a,\alpha}\bra{a,\alpha'}
\right\|_{\infty} \nonumber\\
& & \qquad+\frac{1}{\sqrt{d}}   \left\|\sum_{\alpha,\alpha'}\sqrt{\tr( \tau_{B'}^{a, \alpha} )\tr( \tau_{B'}^{a, \alpha'}) } \ket{a,\alpha}\bra{a,\alpha'}\right\|_{\infty}\nonumber\\
&\leq&   (1        + \frac{1}{\sqrt{d}} )\max_a\left\|      \sum_{\alpha, \alpha'}      \sqrt{\tr( \tau_{B'}^{a, \alpha} )\tr( \tau_{B'}^{a, \alpha'}) }    \vert a, \alpha\rangle \langle a, \alpha'\vert      \right\|_{\infty}\nonumber\\
& = &(1+\frac{1}{\sqrt{d}} )\max_a \sum_\alpha \tr(\tau^{a,\alpha}_{B'})=1+\frac{1}{\sqrt{d}}\ .
\label{eq:bbound}
\end{eqnarray}
Combining~\eqref{eq:abound} and~\eqref{eq:bbound} yields~\eqref{eq:normineq}.
\end{proof}

This result may be expressed in terms of the diamond norm. Let $\{U_a\}_{a=0}^{d}$ be the set of unitaries that rotate the standard basis into a set of such mutually unbiased basis, i.e., $U_a\ket{x}=\ket{e^a_x}$ and let
\begin{equation*}
U_{AB}=\sum_{a=0}^d \proj{a}\otimes U^\dagger_a\ .
\end{equation*}
Then
\begin{equation}
\cO_{AB}=\{O_x=(\id_A\otimes\proj{x}_B)U_{AB}\}_{x=1}^d\ . \label{eq:difficultpovm}
\end{equation}
 is a POVM that exactly outputs $x$ under a randomly chosen $\rho_{AB}^x$. 

\begin{corollary}[Impossibility of instantaneous measurement]\label{thm:instantaneousimpossibility}
Let $\cM_{AB}=\{M_{x}\}_{x}$ be an instantaneous measurement 
on $\mathbb{C}^d\otimes\mathbb{C}^d$ 
implemented with shared entanglement~$\eta_{A'B'}$ as in Fig.~\ref{fig:instmeasurement}. If $\dim B'\leq \varepsilon \sqrt{d}$, then
\begin{equation*}
\|\cM_{AB}-\cO_{AB}\|_\diamond\geq 2(1- 2\varepsilon)\ , 
\end{equation*}
where $\cO_{AB}$ is the POVM~\eqref{eq:difficultpovm}.
\end{corollary}
\begin{proof}
Define 
\begin{equation*}
\rho_{XAB}=\frac{1}{d}\sum_{x}\proj{x}\otimes \rho_{AB}^x\ .
\end{equation*}
Then due to the definition of the diamond norm we have
\begin{eqnarray*}
\|\cM_{AB}-\cO_{AB}\|_\diamond & \geq \| \cI_X\otimes \cM_{AB} (\rho_{XAB})  - \cI_X\otimes \cO_{AB} (\rho_{XAB})  \|_1\\
& = 2(1- p_{\text{succ}})\\
& \geq 2(1-2\varepsilon)\ ,
\end{eqnarray*}
where $p_{\text{succ}}$ denotes the probability that $\cM_{AB}$ successfully finds $x$, and in the last line we use Theorem~\ref{thm:guessingprobbound}.
\end{proof}

\section{Implications for position-based quantum cryptography}
Position-based cryptography revolves around the idea of using
 the geographical location of an entity as its only credential. One of the most basic position-based primitives is position-verification: Here a prover~$P$  tries to convince 
a set of verifiers~$\{V_i\}_{i=1}^S$ that he is at a specific location~$\vec{r}_0\in\mathbb{R}^D$ in space ($D=3$ is of most interest in practice, of course). Proposed protocols for this problem usually assume that the verifiers are located throughout space at different positions~$\vec{r}_{V_i}$ (such that~$\vec{r}_0$ is e.g., in the convex hull of $\{\vec{r}_{V_i}\}_i$)  and have synchronized clocks. They proceed by applying distance bounding techniques~\cite{bc94}: the verifiers send challenges to~$P$ and obtain bounds on their distance to~$P$ by measuring the time taken for his responses to arrive. All computational processes are assumed to be essentially instantaneous, that is, fast in comparison to the ratio between the desired spatial resolution and the speed of light.

A number of classical protocols for position-verification have been proposed in the past~(see e.g.,~\cite{buhrmanetal10} for a list of references), but it was shown in~\cite{chandranetal09} that security is unachievable by any classical protocol without additional assumptions if there are colluding adversaries (at different locations).
 Position-based {\em quantum} cryptography attempts to overcome this impossibility  using  quantum cryptographic techniques. Its history is somewhat reminiscent of the study of bit commitment: this primitive was also known to be classically unachievable and the use of quantum cryptographic techniques appeared promising for a while. Then, Mayers~\cite{mayers97}, and Lo and Chau~\cite{lochau97} showed that secure bit commitment is generally unachievable in a quantum setting. We refer to~\cite{kentetal10} for a nice introduction to the basic ideas of position-based quantum cryptography with examples of simple attacks on specific schemes. Further attacks on previously proposed schemes were found by Lau and Lo in~\cite{laulo10}, and a general impossibility proof was given by Buhrman et al.~\cite{buhrmanetal10} based on Vaidman's techniques for instantaneous measurements. 

The impossibility of secure position-based quantum cryptography has motivated the search for schemes secure under additional assumptions. Since known attacks rely on entanglement shared between colluding adversaries, it is natural to limit the amount of entanglement available to adversaries. In~\cite{buhrmanetal10}, a position-based scheme was shown to be secure under the assumption that the adversaries share no entanglement. In~\cite{laulo10}, a security proof for a  scheme was given that relies on the assumption that the adversaries share entanglement in the form of a two- or three-level system.

Our results on instantaneous quantum computation tighten the impossibility result of~\cite{buhrmanetal10} by reducing the amount of entanglement required for a successful attack to an exponential (instead of doubly exponential) amount. On the positive side, we provide a single-round protocol  which is secure if the adversaries share less entanglement than a given linear bound and are restricted to classical communication. We also relate the security of a  protocol in the setting of no prior shared entanglement to its security  in the case where the adversaries have a limited amount of entanglement (and are otherwise unrestricted). This implies for example that the multi-round protocol of~\cite{buhrmanetal10} remains secure even if the adversaries share less than a linear amount of entanglement.

To illustrate the main points, we  focus on the problem of position-verification in $D=1$~dimension. We assume that the (honest) prover is at a location~$r_0$ between the two verifiers~$V_1$ and~$V_2$, i.e., $r_{V_1}<r_0<r_{V_2}$ and wants to convince them of this fact.    We consider protocols which satisfy the following correctness condition: if $P$ is located at~$r_0$, he can always make the verifiers accept.   A colluding set of adversaries~$\{\tilde{P}_i\}_{i=1}^M$  may try to convince the verifiers that at least one of them is at location~$r_0$, while in fact, none of them actually is, $|r_{\tilde{P}_i}-r_0|> \Delta$ (where $\Delta$ is the desired spatial resolution). The figure of merit is the soundness of such  a protocol. Following~\cite{buhrmanetal10}, we say that a protocol for position-verification is $\varepsilon$-sound if any colluding set of adversaries~$\{\tilde{P}_i\}_{i=1}^M$ cannot make the verifiers accept with probability more than~$\varepsilon$.

For concreteness, we discuss the following protocol $\pi_n$ for position-verification. It depends on a parameter~$n$ and involves, as in Section~\ref{sec:lowerbound}, a set of unitaries~$\{U_a\}_{a=1}^{d}$  taking the computational basis of~$\mathbb{C}^d$ to $d$~mutually unbiased bases, where~$d=2^n$.

\begin{protocol}{$\pi_n(r_0)$}{position-verification proving that $P$ is  at $r_0\in\mathbb{R}$}{

$V_0$ and $V_1$ are at positions $r_{V_0}<r_0<r_{V_1}$ and share common (secret) randomness in the form of uniformly distributed bit strings $a,x\in \sbin^n$. }\label{proto:positionbased}
 
\item[1.] $V_0$ sends~$a$ to~$P$ and $V_1$ prepares the state $U_a\ket{x}$ and sends it to~$P$. The timing is chosen such that both the classical information and the quantum state arrive at~$r_0$ at the same time.
\item[2.] $P$ measures the state in the basis $\{\ket{e^a_i}\}_i$, getting measurement outcome $\hat{x}\in\sbin^n$. He sends $\hat{x}$ to both $V_0$ and~$V_1$.
\item[3.] $V_0$ and $V_1$ accept if they receive $\hat{x}$ at times consistent with~$\hat{x}$ being emitted from~$r_0$ in both directions simultaneously (instantaneously after~$a$ and~$U_a\ket{x}$ reach~$r_0$), and $\hat{x}=x$. 
\end{protocol}

It is easy to check that this protocol is correct in the sense defined above. This and similar protocols have been studied in a series of papers~\cite{kent10,kentetal10,laulo10,malaney10,buhrmanetal10}. The motivation for considering such protocols originated from the no-cloning principle. It seems to suggest that to successfully convince the verifiers, some cheating adversaries have to clone the (classical) information~$x$ without knowledge of the basis this information is encoded in. It was realized early on, however, that this intuition is misleading because entanglement between colluding adversaries can render such schemes insecure. That is, consider two colluding adversaries $\tilde{P}_1,\tilde{P}_2$  located on either side of~$r_0$, with  $r_{V_1}<r_{\tilde{P}_1}<r_0<r_{\tilde{P}_2}<r_{V_2}$. We assume that  their distance $|r_{\tilde{P}_i}-r_0|>\Delta$ from $r_0$ is larger than the desired spatial resolution~$\Delta$. We further assume that they share entanglement. Specializing the results of~\cite{buhrmanetal10} to the protocol~$\pi_{n}(r_0)$ and using our version of instantaneous computation then gives
\begin{lemma}\label{lem:insecurity}
Protocol~$\pi_{n}(r_0)$ is not $\varepsilon$-sound against $\tilde{P}_0$ and $\tilde{P}_1$  if they share~$n\left(1+\frac{2^{8n+5}}{(1-\varepsilon)^2}\right)$ ebits of entanglement.
\end{lemma}
\begin{proof}
To attack the protocol,~$\tilde{P}_0$ and~$\tilde{P}_1$ are faced with the analogous challenge as Alice and Bob in the situation described in the proof of Theorem~\ref{thm:guessingprobbound} (with~$d=2^n$ instead of~$d+1$~mutually unbiased bases). The soundness parameter~$\varepsilon$ is equal to the success probability of correctly guessing~$x$ based on an instantaneous measurement of their state~$\rho^{x}_{AB}$. Since the  non-local POVM~$\cO$ defined by~\eqref{eq:difficultpovm}  does so with certainty, it suffices to approximate this POVM by an instantaneous measurement with accuracy~$1-\varepsilon$  in diamond norm. The claim therefore follows from Theorem~\ref{thm:maintheorem}~\eqref{it:instmeasurement}.
\end{proof}
In contrast, the results of~\cite{buhrmanetal10} rely on Vaidman's techniques, and therefore only show insecurity if~$\tilde{P}_0$ and~$\tilde{P}_1$ share a
doubly exponential (in~$n$) amount of entanglement. The same proof technique applies to more general protocols, including multiple parties, see~\cite{buhrmanetal10}.

On the positive side, we prove security of the protocol~$\pi_{n}(r_0)$ against adversaries limited as follows: they share fewer than $n/2$ ebits of entanglement and are restricted to classical communication.
\begin{lemma}\label{lem:pinsecurityprotocol}
For $n>1$, the protocol $\pi_{n}(r_0)$ is $2\cdot 2^{m-n/2}$-sound against adversaries $\tilde{P}_0$ and $\tilde{P}_1$ sharing at most~$m$ ebits of entanglement and communicating only classical information.
\end{lemma}
\begin{proof}
As in the proof of Lemma~\ref{lem:insecurity}, we can upper bound the
soundness of the protocol~$\pi_{n}(r_0)$ by  Alice and Bob's success probability 
in the setting of Theorem~\ref{thm:guessingprobbound}.
Adapting the proof of the latter shows that 
the upper bound~\eqref{eq:predictionprob} holds even if $d=2^n$ instead of $d+1$~mutually unbiased bases are used. The claim follows immediately.
\end{proof}
The protocol  $\pi_{n}(r_0)$ may not be very practical if $n$ is large because it requires even honest parties to manipulate $n$~qubits at a time. It has, however, the appealing feature that it consists of only one round of communication between the verifiers and the prover and nevertheless achieves exponential security. In the setting where quantum communication is allowed, we may ask whether and how this degree of security could be achieved with a similar one-round protocol. Answering this may prove challenging because of the locking effect~\cite{DiVincenzo04}.

To obtain a more practical  protocol with exponential security, we can consider multi-round protocols: sequential composition decreases the  soundness error of a protocol. For example, it is shown in~\cite{buhrmanetal10} that
the protocol $\pi:=\pi_{1}(r_0)$ is $\varepsilon$-sound 
(even when allowing quantum communication) with 
\begin{equation}
\varepsilon=\varepsilon(\pi_{1}(r_0),\emptyset)<1-h^{-1}(1/2)<1\ ,\label{eq:epsilonsinglebound}
\end{equation}
if there is no shared entanglement between $\tilde{P}_0$ and $\tilde{P}_1$. Here we write~$\emptyset$ to signify absence of entanglement and~$h(p)=-p\log p-(1-p)\log(1-p)$ is the binary entropy function. (We point out that similar  
results were found in~\cite{laulo10} for a constant amount of prior entanglement, that is, qubits or qutrits.)
Equation~\eqref{eq:epsilonsinglebound} implies (see~\cite[Corollary 2]{buhrmanetal10}) that the $L$-fold sequential composition $\pi^{\circ L}$ has exponentially small soundness error
\begin{equation}
\varepsilon(\pi^{\circ L},\emptyset)\leq \varepsilon(\pi,\emptyset)^L< (1-h^{-1}(1/2))^L\ \label{eq:sequentialcompositionL}
\end{equation}
if the adversaries share no entanglement and are restricted to classical communication between rounds. The restriction to classical communication is necessary as the adversaries could otherwise distribute and use an arbitrary amount of entanglement in the course of attacking the composed protocol. 

We now show that  sequential composition of a protocol can provide exponential security even in a setting where the adversaries have a linear  amount of entanglement. This is based on the following relation between the no-entanglement setting and the case of a limited  amount of entanglement.
\begin{lemma}\label{lem:reductionpos}
Consider a protocol~$\pi$ achieving position-verification with soundness error~$\varepsilon(\pi,\emptyset)$ in the case where the adversaries share no entanglement. In a setting where the adversaries  share an (arbitrary) entangled state~$\eta_{A'B'}$, its soundness error $\varepsilon(\pi,\eta_{A'B'})$ is upper bounded as follows:
\begin{equation*}
\varepsilon(\pi,\eta_{A'B'})\leq \dim A'\dim B'\cdot \varepsilon(\pi,\emptyset)\ .
\end{equation*}
\end{lemma}
\begin{proof}
Without loss of generality, we may assume that the protocol~$\pi$ takes the following form. In a first step,  a quantum state~$\rho_{ABST}$  is distributed, with the adversaries  $\tilde{P}_0$, $\tilde{P}_1$ holding~$A$ and $B$, respectively, and the verifiers~$V_0$ and~$V_1$ holding~$S$ and~$T$. This step is followed by various (arbitrary) actions of the adversaries alternating with actions of the verifiers as prescribed by the protocol. We can describe this by a CPTP map~$\cF$ acting on the initial state and the entanglement shared by the adversaries. Finally, a binary valued measurement $\{E,\id-E\}$ (where $0\leq E\leq \id$) is applied to the resulting state. Its outcome determines whether or not the verifiers accept.
  The soundness error of a protocol~$\pi$ when $\tilde{P}_0$ and $\tilde{P}_1$ have a shared entangled state $\eta_{A'B'}$ can then be expressed as
\begin{equation}
\varepsilon(\pi,\eta_{A'B'})=\max_{\cF}\tr (E\cF(\rho_{ABST}\otimes\eta_{A'B'}))\ .\label{eq:fupperbound}
\end{equation}
Here  the maximum is over all CPTP maps compatible with the protocol, i.e., resulting from combining an arbitrary cheating strategy of the adversaries with the fixed actions of the verifiers.

We use the same kind of proof strategy as before  (cf.~\eqref{eq:dimensionboundentangled}), replacing shared entanglement by the completely mixed state. Let $\cF^*$ be the optimal POVM achieving the maximum on the rhs.~of~\eqref{eq:fupperbound}. Then
\begin{eqnarray*}
\varepsilon(\pi,\eta_{A'B'})&=\tr (E\cF^*(\rho_{ABST}\otimes\eta_{A'B'}))\\
&\leq \tr (E\cF^*(\rho_{ABST}\otimes\id_{A'B'}))\\
&=\dim A'\dim B'\tr (E\cF^*(\rho_{ABST}\otimes\theta_{A'B'}))
\end{eqnarray*}
where $\theta_{A'B'}$ is the completely mixed state on $A'B'$. But this expression is equivalent to simply having shared randomness between the adversaries. This implies that
\begin{equation*}
\tr (E\cF^*(\rho_{ABST}\otimes\theta_{A'B'}))\leq \varepsilon(\pi,\emptyset)\ 
\end{equation*}
and the claim follows.
\end{proof}
Combining~\eqref{eq:sequentialcompositionL} with Lemma~\ref{lem:reductionpos}, we obtain the following statement.
\begin{lemma}\label{lem:pionecomposition}
Consider two adversaries sharing an arbitrary entangled state $\eta_{A'B'}$ on~$(\mathbb{C}^2)^{\otimes k}\otimes (\mathbb{C}^2)^{\otimes k}$ and restricted to classical communication.  The $L$-fold sequential composition $\pi_1(r_0)^{\circ L}$ of the protocol $\pi_1(r_0)$ is $\varepsilon$-sound if
\begin{equation*}
L\geq \frac{2k+\log 1/\varepsilon}{\log 1/\delta}\qquad\textrm{ where }\qquad \delta=1-h^{-1}(1/2)\ .
\end{equation*}
\end{lemma}
Lemma~\ref{lem:pionecomposition} and  Lemma~\ref{lem:pinsecurityprotocol} give roughly the same scaling for the maximal amount of tolerable prior entanglement as a function of the soundness parameter. However, the protocol  of Lemma~\ref{lem:pionecomposition} is arguably more practical.

\myacknowledgments

This work was done while the authors were at the Institute for Quantum Information, Caltech. SB acknowledges support by NSF under Grant No.~PHY-0803371 and by NSA/ARO under Grant No.~W911NF-09-1-0442. RK acknowledges support by SNF under grant no.~PA00P2-126220.

\appendix
\section{Derivation of port-based teleportation\label{app:portbased}}
In this appendix, we give a proof of Theorem~\ref{thm:portbasedmain} by Ishizaka and Hiroshima~\cite{ishietal08}. We mostly follow~\cite{ishietal08} and first show the equivalence of port-based teleportation to the problem of distinguishing a certain set of states in Appendix~\ref{sec:qserelation}.  In Appendix~\ref{sec:qsesolved}, we present a general lower bound on the success probability of the so-called pretty good measurement~\cite{hausladenwootters94}. Combining these two facts in Appendix~\ref{sec:portbasedproof} we obtain Theorem~\ref{thm:portbasedmain}.

\subsection{From port-based teleportation to distinguishing quantum states\label{sec:qserelation}}
The crucial observation made by Ishizaka and Hiroshima is the fact that the fidelity of port-based teleportation achieved by  a certain POVM~$\{E^i_{AA'^N}\}_{i=1}^N$ is directly related to a quantum hypothesis testing problem: It is a simple function of the average success probability when using the POVM to distinguish a certain ensemble of states~$\{\eta^i_{AA'^N}\}_{i=1}^N$ with equal prior probabilities. The following lemma expresses this fact.

\begin{lemma}[Equivalence of port-based teleportation with hypothesis testing~\cite{ishietal08}]\label{lem:portbasedequivalence}
Let $\ket{\Phi}=\frac{1}{\sqrt{d}}\sum_{i=1}^d \ket{i}\ket{i}$ be the maximally entangled state.
Let $AA'^NBB'^N\cong \mathbb{C}^d\otimes (\mathbb{C}^d)^N\otimes \mathbb{C}^d\otimes (\mathbb{C}^d)^N$ and define the states
\begin{equation}
\eta^i = \tr_{B'^N\backslash B'_i} \proj{\Phi}^{\otimes N}_{A'^NB'^N}\label{eq:familyofstatestodistinguish}
\end{equation}
on $A'^N\otimes B_i\cong A'^N\otimes B\cong A'^N\otimes A$ for  $i=1,\ldots, N$
using the canonical isomorphism between $A$ and $B$. Consider a POVM
 $\{E_i:=E^i_{AA'^N}\}_{i=1}^N$ on $AA'^N\cong \mathbb{C}^d\otimes (\mathbb{C}^d)^N$ and let 
\begin{equation*}
p_{\text{succ}}=\frac{1}{N}\sum_{i=1}^N \tr(E^i\eta^i)\ 
\end{equation*}
be the average probability of successfully distinguishing the states $\{\eta^i\}_{i=1}^N$
 with uniform prior distribution using this POVM. Let $\cE=\cE_{\ket{\Phi}^{\otimes N}}:\cB(A)\rightarrow\cB(B)$ be the port-based teleportation map~\eqref{eq:multioutputchannel} associated with the POVM. The entanglement fidelity of this CPTP map satisfies
\begin{equation*}
F(\cE)=\frac{N}{d^2}\, p_{\text{succ}}\ .
\end{equation*}
\end{lemma}
\begin{proof}
Fix an orthonormal basis for systems $A\cong B\cong C\cong A'_i\cong B'_i \cong \mathbb{C}^d$. We write $E_A=F_B$ for two operators $E_A$ and $F_B$ acting on isomorphic Hilbert spaces~$A$ and~$B$ if their matrix elements in the computational basis coincide.

The entanglement fidelity of $\cE$ is equal to 
\begin{equation*}
F(\cE) = \tr P_{BC}(\cE\otimes\cI_C)(P_{AC})\ ,
\end{equation*}
 where $P=\proj{\Phi}$ is the projection onto the maximally entangled state~$\ket{\Phi}=\frac{1}{\sqrt{d}}\sum_{i=1}^d\ket{i}\ket{i}$. 
 Omitting identity operators and tensor products for ease of notation, we have 
\begin{eqnarray*}
F(\cE) &= \sum_{i=1}^N\tr P_{BC} \tr_{B'^N\backslash B'_i}\tr_{AA'^N}E^i_{AA'^N}P_{AC}\proj{\Phi}^{\otimes N}_{A'^NB'^N}\\
&=\sum_{i=1}^N\tr P_{BC}  \tr_{AA'^N} E^i_{AA'^N}P_{AC}\eta^{i}_{A'^NB} \\
& = \sum_{i=1}^N\tr P_{BC}   E^i_{AA'^N}P_{AC}\eta^{i}_{A'^NB} \ .
\end{eqnarray*}
For every operator $X$ we have
\begin{equation}
(X\otimes \id)P=(\id \otimes X^T)P\qquad\textrm{ and }\qquad P(\id\otimes X)=P(X^{T}\otimes \id)\ ,
\end{equation}
where $X^T$ is the transpose of matrix $X$ (with respect to the standard basis). 
Then
\begin{eqnarray}
P_{BC}E^i_{AA'^N}P_{AC} & =P_{BC}({E}^i_{CA'^N})^{T_C} P_{AC}\nonumber\\
& =P_{BC}\left((E^i_{BA'^N})^{T_B}\right)^{T_B} P_{AC}\nonumber\\
&=P_{BC} E^i_{BA'^N} P_{AC}\label{eq:pbcac}
\end{eqnarray}
Reinserting~\eqref{eq:pbcac} and using the fact that the partial trace of $P$ is the fully mixed state we obtain
\begin{equation*}
F(\cE) =\frac{1}{d^2} \sum_{i=1}^N \tr\left(  E^i_{A'^NB}\eta^i_{A'^NB}\right) = \frac{N}{d^{2}}\, p_{\text{succ}}\ .
\end{equation*}
\end{proof}

\subsection{A lower bound on the success probability of the pretty good measurement\label{sec:qsesolved}}
Having reduced port-based teleportation to a hypothesis testing problem, it remains to show that there is a suitable POVM solving the latter. Concretely, we need to provide a POVM~$\{E^i\}_{i=1}^N$ that distinguishes the family of states~$\{\eta^i\}_{i=1}^N$, defined by~\eqref{eq:familyofstatestodistinguish}. Ishizaka and Hiroshima~\cite{ishietal08,ishietal09} show that the pretty good measurement
\begin{equation}
E^i =\left(\sum_j\eta^j\right)^{-1/2}\eta^i \left(\sum_k\eta^k\right)^{-1/2}\qquad\textrm{ for }i=1,\ldots, N\label{eq:defpgm}
\end{equation}
does so with sufficiently high success probability
\begin{equation}
p_{\text{succ}}^{\text{pgm}} =\frac{1}{N}\sum_{i=1}^N \tr(E^i\eta^i)\ .\label{eq:pgmsuccess}
\end{equation}
Here we rederive and generalize their bound on the quantity~\eqref{eq:pgmsuccess}: we derive a general lower bound on the success probability of the pretty good measurement for \emph{any} (uniform) family of states~$\{\frac{1}{N},\eta^i\}_{i=1}^N$ (see Lemma~\ref{lem:pgmperformance} below). We will apply this bound to port-based teleportation in Appendix~\ref{sec:portbasedproof}.

Our main technical tool is the following inequality.
\begin{lemma}\label{lem:operatorbound}
Let $X$ and $Y$ be non-negative operators on two (not necessarily identical) Hilbert spaces, satisfying
\begin{equation*}
\tr X=\tr\sqrt{Y}\ .
\end{equation*}
Then 
\begin{equation}
\tr X^2 \geq \frac{(\tr Y)^3}{\mathsf{rank} X \tr Y^2}\ . \label{eq:mainbound}
\end{equation}
\end{lemma}
\begin{proof}
For every non-negative operators $X, Y$ and $\alpha, \gamma>0$, we have
\begin{eqnarray} 
\alpha \cdot \tr X^2&\geq \frac{\alpha^{1/2}}{\sqrt{\mathsf{rank} X }}2\tr X-1\ ,\label{eq:opinequalityone}\\
\gamma\cdot 2\tr\sqrt{Y}&\geq 3-\frac{1}{\gamma^2(\tr Y)^3}\cdot\tr Y^2\ .\label{eq:opinequalitytwo}
\end{eqnarray}
These inequalities follow easily from 
\begin{eqnarray*}
X^2 &\geq 2X-\id_{\text{supp} X}\ ,\\
2\sqrt{Y}&\geq  3Y-Y^2\ , 
\end{eqnarray*}
and rescaling. (Such inequalities have previously been used e.g., in~\cite[Eq.~(32)]{hausladenetal96} for channel coding). Using~$\tr X = \tr \sqrt{Y}$ we obtain
\begin{equation*} 
\tr X^2 \geq \frac{2}{\alpha}-\frac{\mathsf{rank} X}{\alpha^2(\tr Y)^3}\tr Y^2\ ,
\end{equation*}
and the claim follows by letting $\alpha=\left( \frac{(\tr Y)^3}{\mathsf{rank} X \tr Y^2}\right)^{-1}$.
\end{proof}
The proof of the lower bound on the success probability of the pretty good measurement is now straightforward:
\begin{lemma}\label{lem:pgmperformance}
Consider an ensemble  of states $\{\frac{1}{N},\eta^i\}_{i=1}^N$. The success probability~\eqref{eq:pgmsuccess} of the pretty good measurement is bounded by
\begin{equation*}
p_{\text{\emph{succ}}}^{\text{\emph{pgm}}} \geq \frac{1}{N\bar{r}\tr\bar{\eta}^2}
\end{equation*}
where
\begin{equation*}
\bar{\eta} =\frac{1}{N}\sum_{i=1}^N \eta^i\ \textrm{ and }\  \bar{r}=\frac{1}{N}\sum_{i=1}^N \mathsf{rank}\eta^i \ .
\end{equation*}
\end{lemma}
\begin{proof}
Define the (unnormalized) operators
\begin{eqnarray*}
\rho_{IQ}&= \sum_{i=1}^N \proj{i}_I\otimes \eta^i\, \\
Y& = \tr_I \rho_{IQ}\ ,\\
X &=(\id_I\otimes Y)^{-1/4} \rho_{IQ}(\id_I\otimes Y)^{-1/4}\ .
\end{eqnarray*}
It is easy to check that
\begin{eqnarray*}
p_{\text{succ}}^{\text{pgm}}&=\frac{1}{N}\tr X^2\\
\tr X&= \tr(\tr_{I} X)=\tr \sqrt{Y}\ .
\end{eqnarray*}
The claim then follows from Lemma~\ref{lem:operatorbound} using
$\mathsf{rank} (Y^{-1/4}\eta^i Y^{-1/4})=\mathsf{rank}\eta^i$.
\end{proof}

\subsection{Proof of Theorem~\ref{thm:portbasedmain}\label{sec:portbasedproof}}
To prove Theorem~\ref{thm:portbasedmain},  we combine the reformulation of Lemma~\ref{lem:portbasedequivalence} with the lower bound on the success probability of the pretty good measurement (Lemma~\ref{lem:pgmperformance}). 

According to Lemma~\ref{lem:portbasedequivalence}, we need to consider the problem of distinguishing the states
\begin{equation}
\eta^i = \proj{\Phi}_{A_i'B}\otimes \left(\rho^{\otimes (N-1)}\right)_{A'^N \backslash A'_i}\ ,\label{eq:translationinvariance}
\end{equation}
where $\rho=\id/d$ is the completely mixed state. We have 
\begin{equation*}
\tr \left(\proj{\Phi}_{A'_1B}\otimes\rho_{A_2'}\right)\left(\rho_{A_1'}\otimes\proj{\Phi}_{A_2'B}\right) =1/d^3\ ,
\end{equation*}
and then 
\begin{eqnarray*}
\tr(\eta^i)^2 &=1/d^{N-1}\\
\tr\eta^i\eta^j &=\textfrac{1}{d^3}\cdot\textfrac{1}{d^{N-2}}\ ,\qquad\textrm{ for }i\neq j\ .
\end{eqnarray*}
As a result,
\begin{equation*}
\tr\bar{\eta}^2 =\frac{1}{Nd^{N-1}}+\frac{N-1}{Nd^{N+1}}\end{equation*}
for  the ensemble average~$\bar{\eta}=\frac{1}{N}\sum_i \eta^i$.
Furthermore, we have $\bar{r}=\mathsf{rank} \eta^1=d^{N-1}$.
Using Lemma~\ref{lem:pgmperformance}, we conclude that
\begin{eqnarray*}
p_{\text{succ}}^{\text{pgm}} &\geq  \frac{1}{Nd^{N-1}\left(\frac{1}{Nd^{N-1}}+\frac{N-1}{Nd^{N+1}}\right)}\\
&=\frac{d^2}{N}\left(\frac{1}{1+\frac{d^2-1}{N}}\right)\\
&\geq \frac{d^2}{N}\left(1-\frac{d^2-1}{N}\right)\ .
\end{eqnarray*}
In particular, according to Lemma~\ref{lem:portbasedequivalence}, this implies that there is a POVM~$\{E^i_{AA'^N}\}_{i=1}^N$ such that the associated CPTP map~$\cE_{\ket{\Phi}^{\otimes N}}$ achieves port-based teleportation with entanglement fidelity
\begin{equation*}
F(\cE_{\ket{\Phi}^{\otimes N}})=\frac{N}{d^2}\, p_{\text{succ}}^{\text{pgm}}\geq 1-\frac{d^2-1}{N}\ ,
\end{equation*}
as claimed.

 \bibliography{q}

\begin{thebibliography}{10}%
\makeatletter
\providecommand \@ifxundefined [1]{%
 \ifx #1\undefined \expandafter \@firstoftwo
 \else \expandafter \@secondoftwo
\fi
}%
\providecommand \@ifnum [1]{%
 \ifnum #1\expandafter \@firstoftwo
 \else \expandafter \@secondoftwo
\fi
}%
\providecommand \enquote [1]{``#1''}%
\providecommand \bibnamefont  [1]{#1}%
\providecommand \bibfnamefont [1]{#1}%
\providecommand \citenamefont [1]{#1}%
\providecommand\href[0]{\@sanitize\@href}%
\providecommand\@href[1]{\endgroup\@@startlink{#1}\endgroup\@@href}%
\providecommand\@@href[1]{#1\@@endlink}%
\providecommand \@sanitize [0]{\begingroup\catcode`\&12\catcode`\#12\relax}%
\@ifxundefined \pdfoutput {\@firstoftwo}{%
 \@ifnum{\z@=\pdfoutput}{\@firstoftwo}{\@secondoftwo}%
}{%
 \providecommand\@@startlink[1]{\leavevmode\special{html:<a href="#1">}}%
 \providecommand\@@endlink[0]{\special{html:</a>}}%
}{%
 \providecommand\@@startlink[1]{%
  \leavevmode
  \pdfstartlink
   attr{/Border[0 0 1 ]/H/I/C[0 1 1]}%
   user{/Subtype/Link/A<</Type/Action/S/URI/URI(#1)>>}%
  \relax
 }%
 \providecommand\@@endlink[0]{\pdfendlink}%
}%
\providecommand \url  [0]{\begingroup\@sanitize \@url }%
\providecommand \@url [1]{\endgroup\@href {#1}{\urlprefix}}%
\providecommand \urlprefix [0]{URL }%
\providecommand \Eprint[0]{\href }%
\@ifxundefined \urlstyle {%
  \providecommand \doi [1]{doi:\discretionary{}{}{}#1}%
}{%
  \providecommand \doi [0]{doi:\discretionary{}{}{}\begingroup
  \urlstyle{rm}\Url }%
}%
\providecommand \doibase [0]{http://dx.doi.org/}%
\providecommand \Doi[1]{\href{\doibase#1}}%
\providecommand \bibAnnote [3]{%
  \BibitemShut{#1}%
  \begin{quotation}\noindent
    \textsc{Key:}\ #2\\\textsc{Annotation:}\ #3%
  \end{quotation}%
}%
\providecommand \bibAnnoteFile [2]{%
  \IfFileExists{#2}{\bibAnnote {#1} {#2} {\input{#2}}}{}%
}%
\providecommand \typeout [0]{\immediate \write \m@ne }%
\providecommand \selectlanguage [0]{\@gobble}%
\providecommand \bibinfo [0]{\@secondoftwo}%
\providecommand \bibfield [0]{\@secondoftwo}%
\providecommand \translation [1]{[#1]}%
\providecommand \BibitemOpen[0]{}%
\providecommand \bibitemStop [0]{}%
\providecommand \bibitemNoStop [0]{.\EOS\space}%
\providecommand \EOS [0]{\spacefactor3000\relax}%
\providecommand \BibitemShut [1]{\csname bibitem#1\endcsname}%
\bibitem{epr35}%
  \BibitemOpen
  \bibfield{author}{%
  \bibinfo {author} {\bibfnamefont{A.}~\bibnamefont{Einstein}}, \bibinfo
  {author} {\bibfnamefont{B.}~\bibnamefont{Podolsky}},\ and\ \bibinfo {author}
  {\bibfnamefont{N.}~\bibnamefont{Rosen}},\ }%
  \bibfield{journal}{%
  \Doi{10.1103/PhysRev.47.777}{\bibinfo {journal} {Phys. Rev.}}\ }%
  \textbf{\bibinfo {volume} {47}},\ \bibinfo {pages} {777} (\bibinfo {month}
  {May}\ \bibinfo {year} {1935})%
  \bibAnnoteFile{NoStop}{epr35}%
\bibitem{landaupeierls31}%
  \BibitemOpen
  \bibfield{author}{%
  \bibinfo {author} {\bibfnamefont{L.}~\bibnamefont{Landau}}\ and\ \bibinfo
  {author} {\bibfnamefont{R.}~\bibnamefont{Peierls}},\ }%
  \bibfield{journal}{%
  \bibinfo {journal} {Z. Phys.}\ }%
  \textbf{\bibinfo {volume} {69}} (\bibinfo {year} {1931})%
  \bibAnnoteFile{NoStop}{landaupeierls31}%
\bibitem{clarketal10}%
  \BibitemOpen
  \bibfield{author}{%
  \bibinfo {author} {\bibfnamefont{S.~R.}\ \bibnamefont{Clark}}, \bibinfo
  {author} {\bibfnamefont{A.~J.}\ \bibnamefont{Connor}}, \bibinfo {author}
  {\bibfnamefont{D.}~\bibnamefont{Jaksch}},\ and\ \bibinfo {author}
  {\bibfnamefont{S.}~\bibnamefont{Popescu}},\ }%
  \bibfield{journal}{%
  \bibinfo {journal} {New Journal of Physics}\ }%
  \textbf{\bibinfo {volume} {12}} (\bibinfo {month} {August}\ \bibinfo {year}
  {2010})%
  \bibAnnoteFile{NoStop}{clarketal10}%
\bibitem{bohrrosenfeld33}%
  \BibitemOpen
  \bibfield{author}{%
  \bibinfo {author} {\bibfnamefont{N.}~\bibnamefont{Bohr}}\ and\ \bibinfo
  {author} {\bibfnamefont{L.}~\bibnamefont{Rosenfeld}},\ }%
  \bibfield{journal}{%
  \bibinfo {journal} {Mat.-fys. Medd. Dansk Vid. Selsk.}\ }%
  \textbf{\bibinfo {volume} {12}} (\bibinfo {year} {1933})%
  \bibAnnoteFile{NoStop}{bohrrosenfeld33}%
\bibitem{aharonovalbert80}%
  \BibitemOpen
  \bibfield{author}{%
  \bibinfo {author} {\bibfnamefont{Y.}~\bibnamefont{Aharonov}}\ and\ \bibinfo
  {author} {\bibfnamefont{D.~Z.}\ \bibnamefont{Albert}},\ }%
  \bibfield{journal}{%
  \Doi{10.1103/PhysRevD.21.3316}{\bibinfo {journal} {Phys. Rev. D}}\ }%
  \textbf{\bibinfo {volume} {21}},\ \bibinfo {pages} {3316} (\bibinfo {month}
  {Jun}\ \bibinfo {year} {1980})%
  \bibAnnoteFile{NoStop}{aharonovalbert80}%
\bibitem{aharonovalbert81}%
  \BibitemOpen
  \bibfield{author}{%
  \bibinfo {author} {\bibfnamefont{Y.}~\bibnamefont{Aharonov}}\ and\ \bibinfo
  {author} {\bibfnamefont{D.~Z.}\ \bibnamefont{Albert}},\ }%
  \bibfield{journal}{%
  \Doi{10.1103/PhysRevD.24.359}{\bibinfo {journal} {Phys. Rev. D}}\ }%
  \textbf{\bibinfo {volume} {24}},\ \bibinfo {pages} {359} (\bibinfo {month}
  {Jul}\ \bibinfo {year} {1981})%
  \bibAnnoteFile{NoStop}{aharonovalbert81}%
\bibitem{ahaa84}%
  \BibitemOpen
  \bibfield{author}{%
  \bibinfo {author} {\bibfnamefont{Y.}~\bibnamefont{Aharonov}}\ and\ \bibinfo
  {author} {\bibfnamefont{D.~Z.}\ \bibnamefont{Albert}},\ }%
  \bibfield{journal}{%
  \Doi{10.1103/PhysRevD.29.223}{\bibinfo {journal} {Phys. Rev. D}}\ }%
  \textbf{\bibinfo {volume} {29}},\ \bibinfo {pages} {223} (\bibinfo {month}
  {Jan}\ \bibinfo {year} {1984})%
  \bibAnnoteFile{NoStop}{ahaa84}%
\bibitem{ahab84}%
  \BibitemOpen
  \bibfield{author}{%
  \bibinfo {author} {\bibfnamefont{Y.}~\bibnamefont{Aharonov}}\ and\ \bibinfo
  {author} {\bibfnamefont{D.~Z.}\ \bibnamefont{Albert}},\ }%
  \bibfield{journal}{%
  \Doi{10.1103/PhysRevD.29.228}{\bibinfo {journal} {Phys. Rev. D}}\ }%
  \textbf{\bibinfo {volume} {29}},\ \bibinfo {pages} {228} (\bibinfo {month}
  {Jan}\ \bibinfo {year} {1984})%
  \bibAnnoteFile{NoStop}{ahab84}%
\bibitem{ahaavaid86}%
  \BibitemOpen
  \bibfield{author}{%
  \bibinfo {author} {\bibfnamefont{Y.}~\bibnamefont{Aharonov}}, \bibinfo
  {author} {\bibfnamefont{D.~Z.}\ \bibnamefont{Albert}},\ and\ \bibinfo
  {author} {\bibfnamefont{L.}~\bibnamefont{Vaidman}},\ }%
  \bibfield{journal}{%
  \Doi{10.1103/PhysRevD.34.1805}{\bibinfo {journal} {Phys. Rev. D}}\ }%
  \textbf{\bibinfo {volume} {34}},\ \bibinfo {pages} {1805} (\bibinfo {month}
  {Sep}\ \bibinfo {year} {1986})%
  \bibAnnoteFile{NoStop}{ahaavaid86}%
\bibitem{popescuvaid94}%
  \BibitemOpen
  \bibfield{author}{%
  \bibinfo {author} {\bibfnamefont{S.}~\bibnamefont{Popescu}}\ and\ \bibinfo
  {author} {\bibfnamefont{L.}~\bibnamefont{Vaidman}},\ }%
  \bibfield{journal}{%
  \Doi{10.1103/PhysRevA.49.4331}{\bibinfo {journal} {Phys. Rev. A}}\ }%
  \textbf{\bibinfo {volume} {49}},\ \bibinfo {pages} {4331} (\bibinfo {month}
  {Jun}\ \bibinfo {year} {1994})%
  \bibAnnoteFile{NoStop}{popescuvaid94}%
\bibitem{groismanvaidman01}%
  \BibitemOpen
  \bibfield{author}{%
  \bibinfo {author} {\bibfnamefont{B.}~\bibnamefont{Groisman}}\ and\ \bibinfo
  {author} {\bibfnamefont{L.}~\bibnamefont{Vaidman}},\ }%
  \bibfield{journal}{%
  \bibinfo {journal} {Journal of Physics A: Mathematical and General}\ }%
  \textbf{\bibinfo {volume} {34}},\ \bibinfo {pages} {6881} (\bibinfo {year}
  {2001})%
  \bibAnnoteFile{NoStop}{groismanvaidman01}%
\bibitem{vaidman03}%
  \BibitemOpen
  \bibfield{author}{%
  \bibinfo {author} {\bibfnamefont{L.}~\bibnamefont{Vaidman}},\ }%
  \bibfield{journal}{%
  \Doi{10.1103/PhysRevLett.90.010402}{\bibinfo {journal} {Phys. Rev. Lett.}}\
  }%
  \textbf{\bibinfo {volume} {90}},\ \bibinfo {pages} {010402} (\bibinfo {month}
  {Jan}\ \bibinfo {year} {2003})%
  \bibAnnoteFile{NoStop}{vaidman03}%
\bibitem{buhrmanetal10}%
  \BibitemOpen
  \bibfield{author}{%
  \bibinfo {author} {\bibfnamefont{H.}~\bibnamefont{Buhrman}}, \bibinfo
  {author} {\bibfnamefont{N.}~\bibnamefont{Chandran}}, \bibinfo {author}
  {\bibfnamefont{S.}~\bibnamefont{Fehr}}, \bibinfo {author}
  {\bibfnamefont{R.}~\bibnamefont{Gelles}}, \bibinfo {author}
  {\bibfnamefont{V.}~\bibnamefont{Goyal}}, \bibinfo {author}
  {\bibfnamefont{R.}~\bibnamefont{Ostrovsky}},\ and\ \bibinfo {author}
  {\bibfnamefont{C.}~\bibnamefont{Schaffner}},\ }%
  in\ \emph{\bibinfo {booktitle} {Advances in Cryptology - CRYPTO 2011 - 31st
  Annual Cryptology Conference}},\ Vol.\ \bibinfo {volume} {6841}\ (\bibinfo
  {year} {2011})\ p.\ \bibinfo {pages} {423}%
  \bibAnnoteFile{NoStop}{buhrmanetal10}%
\bibitem{ishietal08}%
  \BibitemOpen
  \bibfield{author}{%
  \bibinfo {author} {\bibfnamefont{S.}~\bibnamefont{Ishizaka}}\ and\ \bibinfo
  {author} {\bibfnamefont{T.}~\bibnamefont{Hiroshima}},\ }%
  \bibfield{journal}{%
  \Doi{10.1103/PhysRevLett.101.240501}{\bibinfo {journal} {Phys. Rev. Lett.}}\
  }%
  \textbf{\bibinfo {volume} {101}},\ \bibinfo {pages} {240501} (\bibinfo
  {month} {Dec}\ \bibinfo {year} {2008})%
  \bibAnnoteFile{NoStop}{ishietal08}%
\bibitem{ishietal09}%
  \BibitemOpen
  \bibfield{author}{%
  \bibinfo {author} {\bibfnamefont{S.}~\bibnamefont{Ishizaka}}\ and\ \bibinfo
  {author} {\bibfnamefont{T.}~\bibnamefont{Hiroshima}},\ }%
  \bibfield{journal}{%
  \Doi{10.1103/PhysRevA.79.042306}{\bibinfo {journal} {Phys. Rev. A}}\ }%
  \textbf{\bibinfo {volume} {79}},\ \bibinfo {pages} {042306} (\bibinfo {month}
  {Apr}\ \bibinfo {year} {2009})%
  \bibAnnoteFile{NoStop}{ishietal09}%
\bibitem{laulo10}%
  \BibitemOpen
  \bibfield{author}{%
  \bibinfo {author} {\bibfnamefont{H.~K.}\ \bibnamefont{Lau}}\ and\ \bibinfo
  {author} {\bibfnamefont{H.~K.}\ \bibnamefont{Lo}},\ }%
  \bibfield{journal}{%
  \bibinfo {journal} {Phys. Rev. A}\ }%
  \textbf{\bibinfo {volume} {83}},\ \bibinfo {pages} {012322} (\bibinfo {month}
  {Jan}\ \bibinfo {year} {2011})%
  \bibAnnoteFile{NoStop}{laulo10}%
\bibitem{Bennett93}%
  \BibitemOpen
  \bibfield{author}{%
  \bibinfo {author} {\bibfnamefont{C.~H.}\ \bibnamefont{Bennett}}, \bibinfo
  {author} {\bibfnamefont{G.}~\bibnamefont{Brassard}}, \bibinfo {author}
  {\bibfnamefont{C.}~\bibnamefont{Cr\'epeau}}, \bibinfo {author}
  {\bibfnamefont{R.}~\bibnamefont{Jozsa}}, \bibinfo {author}
  {\bibfnamefont{A.}~\bibnamefont{Peres}},\ and\ \bibinfo {author}
  {\bibfnamefont{W.~K.}\ \bibnamefont{Wootters}},\ }%
  \bibfield{journal}{%
  \bibinfo {journal} {Phys. Rev. Lett.}\ }%
  \textbf{\bibinfo {volume} {70}},\ \bibinfo {pages} {1895} (\bibinfo {month}
  {Mar}\ \bibinfo {year} {1993})%
  \bibAnnoteFile{NoStop}{Bennett93}%
\bibitem{horodeckietal99}%
  \BibitemOpen
  \bibfield{author}{%
  \bibinfo {author} {\bibfnamefont{M.}~\bibnamefont{Horodecki}}, \bibinfo
  {author} {\bibfnamefont{P.}~\bibnamefont{Horodecki}},\ and\ \bibinfo {author}
  {\bibfnamefont{R.}~\bibnamefont{Horodecki}},\ }%
  \bibfield{journal}{%
  \Doi{10.1103/PhysRevA.60.1888}{\bibinfo {journal} {Phys. Rev. A}}\ }%
  \textbf{\bibinfo {volume} {60}},\ \bibinfo {pages} {1888} (\bibinfo {month}
  {Sep}\ \bibinfo {year} {1999})%
  \bibAnnoteFile{NoStop}{horodeckietal99}%
\bibitem{choi}%
  \BibitemOpen
  \bibfield{author}{%
  \bibinfo {author} {\bibfnamefont{M.~D.}\ \bibnamefont{Choi}},\ }%
  \bibfield{journal}{%
  \bibinfo {journal} {Linear Algebra and its Applications}\ }%
  \textbf{\bibinfo {volume} {10}} (\bibinfo {year} {1975})%
  \bibAnnoteFile{NoStop}{choi}%
\bibitem{jamiolkowski}%
  \BibitemOpen
  \bibfield{author}{%
  \bibinfo {author} {\bibfnamefont{A.}~\bibnamefont{Jamiolkowski}},\ }%
  \bibfield{journal}{%
  \bibinfo {journal} {Reports on Mathematical Physics}\ }%
  \textbf{\bibinfo {volume} {3}},\ \bibinfo {pages} {275} (\bibinfo {year}
  {1972})%
  \bibAnnoteFile{NoStop}{jamiolkowski}%
\bibitem{fuchsgraaf97}%
  \BibitemOpen
  \bibfield{author}{%
  \bibinfo {author} {\bibfnamefont{C.~A.}\ \bibnamefont{Fuchs}}\ and\ \bibinfo
  {author} {\bibfnamefont{J.}~\bibnamefont{van~de Graaf}},\ }%
  \bibfield{journal}{%
  \bibinfo {journal} {IEEE Transactions on Information Theory}\ }%
  \textbf{\bibinfo {volume} {45}},\ \bibinfo {pages} {1216} (\bibinfo {year}
  {1999})%
  \bibAnnoteFile{NoStop}{fuchsgraaf97}%
\bibitem{watrous09}%
  \BibitemOpen
  \bibfield{author}{%
  \bibinfo {author} {\bibfnamefont{J.}~\bibnamefont{Watrous}},\ }%
  \bibfield{journal}{%
  \bibinfo {journal} {Theory of Computing}\ }%
  \textbf{\bibinfo {volume} {5}},\ \bibinfo {pages} {217} (\bibinfo {year}
  {2009})%
  \bibAnnoteFile{NoStop}{watrous09}%
\bibitem{ivanovic81}%
  \BibitemOpen
  \bibfield{author}{%
  \bibinfo {author} {\bibfnamefont{I.~D.}\ \bibnamefont{Ivanovic}},\ }%
  \bibfield{journal}{%
  \bibinfo {journal} {Journal of Physics A}\ }%
  \textbf{\bibinfo {volume} {14}},\ \bibinfo {pages} {3241} (\bibinfo {year}
  {1981})%
  \bibAnnoteFile{NoStop}{ivanovic81}%
\bibitem{woottersfields89}%
  \BibitemOpen
  \bibfield{author}{%
  \bibinfo {author} {\bibfnamefont{W.~K.}\ \bibnamefont{Wootters}}\ and\
  \bibinfo {author} {\bibfnamefont{B.~D.}\ \bibnamefont{Fields}},\ }%
  \bibfield{journal}{%
  \bibinfo {journal} {Annals of Physics}\ }%
  \textbf{\bibinfo {volume} {191}},\ \bibinfo {pages} {363} (\bibinfo {year}
  {1989})%
  \bibAnnoteFile{NoStop}{woottersfields89}%
\bibitem{bandy02}%
  \BibitemOpen
  \bibfield{author}{%
  \bibinfo {author} {\bibfnamefont{S.}~\bibnamefont{Bandyopadhyay}}, \bibinfo
  {author} {\bibfnamefont{P.~O.}\ \bibnamefont{Boykin}}, \bibinfo {author}
  {\bibfnamefont{V.~P.}\ \bibnamefont{Roychowdhury}},\ and\ \bibinfo {author}
  {\bibfnamefont{F.}~\bibnamefont{Vatan}},\ }%
  \bibfield{journal}{%
  \bibinfo {journal} {Algorithmica}\ }%
  \textbf{\bibinfo {volume} {34}},\ \bibinfo {pages} {512} (\bibinfo {year}
  {2002})%
  \bibAnnoteFile{NoStop}{bandy02}%
\bibitem{bc94}%
  \BibitemOpen
  \bibfield{author}{%
  \bibinfo {author} {\bibfnamefont{S.}~\bibnamefont{Brands}}\ and\ \bibinfo
  {author} {\bibfnamefont{D.}~\bibnamefont{Chaum}},\ }%
  \bibfield{journal}{%
  \bibinfo {journal} {EUROCRYPT'93},\ \bibinfo {pages} {344}}%
   (\bibinfo {year} {1994})%
  \bibAnnoteFile{NoStop}{bc94}%
\bibitem{chandranetal09}%
  \BibitemOpen
  \bibfield{author}{%
  \bibinfo {author} {\bibfnamefont{N.}~\bibnamefont{Chandran}}, \bibinfo
  {author} {\bibfnamefont{V.}~\bibnamefont{Goyal}}, \bibinfo {author}
  {\bibfnamefont{R.}~\bibnamefont{Moriarty}},\ and\ \bibinfo {author}
  {\bibfnamefont{R.}~\bibnamefont{Ostrovsky}},\ }%
  in\ \emph{\bibinfo {booktitle} {CRYPTO}}\ (\bibinfo {year} {2009})\ pp.\
  \bibinfo {pages} {391--407}%
  \bibAnnoteFile{NoStop}{chandranetal09}%
\bibitem{mayers97}%
  \BibitemOpen
  \bibfield{author}{%
  \bibinfo {author} {\bibfnamefont{D.}~\bibnamefont{Mayers}},\ }%
  \bibfield{journal}{%
  \Doi{10.1103/PhysRevLett.78.3414}{\bibinfo {journal} {Phys. Rev. Lett.}}\ }%
  \textbf{\bibinfo {volume} {78}},\ \bibinfo {pages} {3414} (\bibinfo {month}
  {Apr}\ \bibinfo {year} {1997})%
  \bibAnnoteFile{NoStop}{mayers97}%
\bibitem{lochau97}%
  \BibitemOpen
  \bibfield{author}{%
  \bibinfo {author} {\bibfnamefont{H.-K.}\ \bibnamefont{Lo}}\ and\ \bibinfo
  {author} {\bibfnamefont{H.~F.}\ \bibnamefont{Chau}},\ }%
  \bibfield{journal}{%
  \Doi{10.1103/PhysRevLett.78.3410}{\bibinfo {journal} {Phys. Rev. Lett.}}\ }%
  \textbf{\bibinfo {volume} {78}},\ \bibinfo {pages} {3410} (\bibinfo {month}
  {Apr}\ \bibinfo {year} {1997})%
  \bibAnnoteFile{NoStop}{lochau97}%
\bibitem{kentetal10}%
  \BibitemOpen
  \bibfield{author}{%
  \bibinfo {author} {\bibfnamefont{A.}~\bibnamefont{Kent}}, \bibinfo {author}
  {\bibfnamefont{B.}~\bibnamefont{Munro}},\ and\ \bibinfo {author}
  {\bibfnamefont{T.}~\bibnamefont{Spiller}},\ }%
  \bibfield{journal}{%
  \bibinfo {journal} {Phys. Rev. A}\ }%
  \textbf{\bibinfo {volume} {84}},\ \bibinfo {pages} {012326} (\bibinfo {month}
  {July}\ \bibinfo {year} {2011})%
  \bibAnnoteFile{NoStop}{kentetal10}%
\bibitem{kent10}%
  \BibitemOpen
  \bibfield{author}{%
  \bibinfo {author} {\bibfnamefont{A.}~\bibnamefont{Kent}},\ }%
  \enquote{\bibinfo {title} {Quantum tagging with cryptographically secure
  tags},}\  (\bibinfo {year} {2010}),\
  \Eprint{http://arxiv.org/abs/arXiv:1008.5380v2}{arXiv:1008.5380v2}%
  \bibAnnoteFile{NoStop}{kent10}%
\bibitem{malaney10}%
  \BibitemOpen
  \bibfield{author}{%
  \bibinfo {author} {\bibfnamefont{R.~A.}\ \bibnamefont{Malaney}},\ }%
  \bibfield{journal}{%
  \Doi{10.1103/PhysRevA.81.042319}{\bibinfo {journal} {Phys. Rev. A}}\ }%
  \textbf{\bibinfo {volume} {81}},\ \bibinfo {pages} {042319} (\bibinfo {month}
  {Apr}\ \bibinfo {year} {2010})%
  \bibAnnoteFile{NoStop}{malaney10}%
\bibitem{DiVincenzo04}%
  \BibitemOpen
  \bibfield{author}{%
  \bibinfo {author} {\bibfnamefont{D.~P.}\ \bibnamefont{DiVincenzo}}, \bibinfo
  {author} {\bibfnamefont{M.}~\bibnamefont{Horodecki}}, \bibinfo {author}
  {\bibfnamefont{D.~W.}\ \bibnamefont{Leung}}, \bibinfo {author}
  {\bibfnamefont{J.~A.}\ \bibnamefont{Smolin}},\ and\ \bibinfo {author}
  {\bibfnamefont{B.~M.}\ \bibnamefont{Terhal}},\ }%
  \bibfield{journal}{%
  \Doi{10.1103/PhysRevLett.92.067902}{\bibinfo {journal} {Phys. Rev. Lett.}}\
  }%
  \textbf{\bibinfo {volume} {92}},\ \bibinfo {pages} {067902} (\bibinfo {month}
  {Feb}\ \bibinfo {year} {2004})%
  \bibAnnoteFile{NoStop}{DiVincenzo04}%
\bibitem{hausladenwootters94}%
  \BibitemOpen
  \bibfield{author}{%
  \bibinfo {author} {\bibfnamefont{P.}~\bibnamefont{Hausladen}}\ and\ \bibinfo
  {author} {\bibfnamefont{W.~K.}\ \bibnamefont{Wootters}},\ }%
  \bibfield{journal}{%
  \bibinfo {journal} {Journal of Modern Optics}\ }%
  \textbf{\bibinfo {volume} {41}},\ \bibinfo {pages} {2385} (\bibinfo {year}
  {1994})%
  \bibAnnoteFile{NoStop}{hausladenwootters94}%
\bibitem{hausladenetal96}%
  \BibitemOpen
  \bibfield{author}{%
  \bibinfo {author} {\bibfnamefont{P.}~\bibnamefont{Hausladen}}, \bibinfo
  {author} {\bibfnamefont{R.}~\bibnamefont{Jozsa}}, \bibinfo {author}
  {\bibfnamefont{B.}~\bibnamefont{Schumacher}}, \bibinfo {author}
  {\bibfnamefont{M.}~\bibnamefont{Westmoreland}},\ and\ \bibinfo {author}
  {\bibfnamefont{W.~K.}\ \bibnamefont{Wootters}},\ }%
  \bibfield{journal}{%
  \Doi{10.1103/PhysRevA.54.1869}{\bibinfo {journal} {Phys. Rev. A}}\ }%
  \textbf{\bibinfo {volume} {54}},\ \bibinfo {pages} {1869} (\bibinfo {month}
  {Sep}\ \bibinfo {year} {1996})%
  \bibAnnoteFile{NoStop}{hausladenetal96}%
\end{thebibliography}%
\end{document}